\documentclass[12pt,draftcls,onecolumn]{ieeeconf}

\IEEEoverridecommandlockouts                              

\usepackage{amsfonts,mathrsfs,graphicx,subfigure,stfloats,amsmath,fancyhdr,amssymb}
\usepackage{diagbox}
\usepackage{booktabs}
\usepackage{color}
\newtheorem{definition}{Definition}[section]
\newtheorem{theorem}{Theorem}[section]
\newtheorem{lemma}{Lemma}[section]
\newtheorem{remark}{Remark}[section]
\newtheorem{assumption}{Assumption}[section]
\newtheorem{corollary}{Corollary}[section]

\title{\LARGE \bf
	Iterative Learning Economic Model Predictive Control}

\author{$\text{Yushen Long}^{1}$, $\text{Lihua Xie}^{1}$ and $\text{Shuai Liu}^{2}$ 
	\thanks{$^{1}$School of Electrical and Electronic Engineering, Nanyang Technological University, Singapore 639798.        {\tt\small \{yslong, elhxie\}@ntu.edu.sg} }
	\thanks{$^{2}$School of Control Science and Engineering, Shandong University, Jinan 250061, China.       {\tt\small liushuai@sdu.edu.cn} }
}

\begin{document}
	\maketitle
	\begin{abstract}
		An iterative learning based economic model predictive controller (ILEMPC) is proposed for repetitive tasks in this paper. Compared with existing works, the initial feasible trajectory of the proposed ILEMPC is not restricted to be convergent to an equilibrium so it can handle various types of control objectives: stabilization, tracking a periodic trajectory and even pure economic optimization. The controller can learn from the previous closed-loop trajectory, resulting in a performance which is guaranteed to be no worse than the previous one. Under some standard assumptions in model predictive control, we show that recursive feasibility is ensured. Furthermore, for stabilization problem, the convergence of each learned trajectory and the learning process are established provided the initial trajectory is convergent. Numerical examples show that the proposed control strategy works well for different types of control tasks and systems.
		
	\end{abstract}
	\section{Introduction}\label{sec1}
	A recent survey in \cite{Samad17} indicates that model predictive control (MPC) is the second most successful control technology in industry and the most successful advanced control technology if one excludes PID control from the advanced ones. A lot of existing literature have contributed to the theoretical analysis \cite{Mayne00}, \cite{Rawlingsbook}, \cite{MAYNE20142967} and practical application \cite{Qin03}, \cite{AFRAM2014343}, \cite{6775336} of MPC.
	
	Note that though the control action of an MPC controller comes from the optimal solution of an open-loop optimal control problem, the resulted closed-loop trajectory is usually different from the optimal closed-loop trajectory. This motivates us to consider if we could improve the closed-loop performance of the system further. The combination of iterative learning control (ILC) and MPC seems to be a promising direction. ILC \cite{Bristow06} is a control strategy for systems which execute the same task multiple times. The performance of the system can be improved by learning from previous executions. ILC usually assumes that for each iteration, the system works under the same initial condition and disturbance realization. The main feature of ILC is that the controller uses information from previous executions to improve the performance of the system, such as minimization of tracking error and rejection of periodic disturbance \cite{JayH07}. The combination of ILC and MPC has been studied in a few papers. In \cite{Kwang99}, the authors propose an ILMPC for batch processes, which is based on a time-varying MIMO linear model. Experiments on a nonlinear batch reactor system show that it outperforms the traditional PID controller and ILC controller. In \cite{KSLee00} the authors prove that for linear systems, the tracking error of the controller in \cite{Kwang99} converges to zero as the number of iterations goes to infinity. In \cite{LEE2000641}, the authors further extend the analysis to linear time-varying systems with disturbances. The tracking errors of previous iterations are explicitly incorporated into the control input of current iteration in order to minimize the tracking error. An observer is also designed when system is subject to deterministic or stochastic disturbances. In \cite{OH2016284}, an ILMPC is formulated based on an incremental state-space model. It is proved that for a disturbance-free linear system, the tracking error converges to zero. An extension to cases with disturbances is also discussed and tested by numerical examples. An ILMPC for nonlinear systems is proposed in \cite{CUELI20081023} based on a series of time-varying linear models along the state trajectory. Assuming that the desired reference is reachable, the authors prove that tracking error converges to zero under some mild assumptions. The non-linearity of system model is handled by a T-Z fuzzy model in \cite{LIU20131023}. The disturbance is rejected by an MPC and an ILC is designed to minimize the accumulative tracking error and excessive input movement. The tracking error along iterations is also proved to be convergent to zero. Different from the aforementioned papers, in \cite{Rosolia17}, the authors do not assume that the reference trajectory is known. A `database' is constructed by using trajectories of previous iterations and the terminal condition of MPC controller is formulated by choosing the best trajectory from the `database'. Under some convexity conditions, the authors prove that if the trajectory converges as iteration index goes to infinity, then the limit trajectory is the optimal solution of a quasi-infinite horizon optimal control problem.
	
	The MPC formulations in the aforementioned references are all trying to minimize the tracking error, with respect to a desired trajectory \cite{KSLee00,LEE2000641,CUELI20081023} or an equilibrium \cite{Rosolia17}. However, such a tracking type cost function does not necessarily represent the actual economic cost involved in plant operation. Economic MPC (EMPC) \cite{Diehl11} has been studied in recent years as a tool to trade-off system behavior between two extreme cases: pure economic optimization and  pure tracking problem. Compared with standard stabilizing MPC, where to ensure the stability, the stage cost must be chosen as a positive definite function with respect to the equilibrium, EMPC allows an arbitrary stage cost function, and hence performance indexes other than tracking error could be handled. As a result, EMPC could be adopted in more real world applications where tracking error is not the main consideration, such as process industry \cite{ELLIS20142561}, water distribution systems \cite{WANG201723} and smart buildings \cite{MA20141282}.
	
	The objective of this paper is to design an ILC algorithm which is not limited to the tracking problem. A designed controller should be able to learn from previous iterations to improve the closed-loop system performance, which is not necessarily the tracking error. Considering the constraints and performance optimization, we combine the iterative learning approach with EMPC formulation and propose an ILEMPC algorithm.
	
	The contribution of this paper is summarized as follows:
	
	1)  We propose a novel ILEMPC algorithm for iterative tasks. Compared with \cite{Rosolia17}, we do not assume that the initial feasible trajectory converges to a steady state. Therefore, more types of control tasks can be handled by our methods. In our MPC controller formulation, no terminal cost function is used, which allows us to handle the situations with infinite accumulative performance index, such as imperfect tracking. Such kind of control task cannot be directly accomplished by the method proposed in \cite{Rosolia17} since the terminal cost used there will be infinite. Furthermore, the terminal constraint in our formulation is a single equality constraint, which is commonly used in existing MPC literature \cite{Mayne90}, \cite{Michalska91}, \cite{Angeli12}, so that the computational complexity of the proposed algorithm is the same as a standard MPC problem. The controller formulation in \cite{Rosolia17}, on the other hand, requires to solve a mixed integer programming at each time instant, which costs significantly more computational resources than a standard one. Finally, we do not assume that the optimum of the stage cost is the steady state, which allows us to optimize economic cost of the plant directly.
	
	2) We show that the recursive feasibility is guaranteed for the proposed MPC controller formulation. For the cases with infinite accumulative performance index, we prove that the average cost of the $j$-th iteration is not worse than that of the $(j-1)$-th iteration. For the cases when the initial feasible trajectory converges to a steady state, we prove that, under similar assumptions in \cite{Angeli12}, the convergence can be preserved for each iteration. Furthermore, the performance improvement along learning process is also guaranteed. Finally, we show that if the iteration converges, then we can obtain the $N$-receding-horizon optimal (see Definition \ref{recedingopt}) trajectory over the infinite horizon, provided some assumptions on the uniqueness of the optimums are met.
	
	The rest of this paper is organized as follows. In Section \ref{formulation} we formulate the iterative control problem, introduce the performance index to be optimized and give the ILEMPC formulation. In Section \ref{analysis}, we present theoretical analysis of the proposed ILEMPC algorithm. An extension to cases with average constraints is investigated in Section \ref{avecon}. In Section \ref{example}, a few numerical examples are given to illustrate the effectiveness of ILEMPC for different types of control tasks. In Section \ref{conclusion}, some conclusions will be drawn. 
	
	Some remarks on notations are introduced as follows. We use $\mathbb{R}$ to denote the set of real numbers. $\mathbb{R}^n$ and $\mathbb{N}$ denote $n$-dimensional Euclidean space and the set of natural numbers, respectively. For a vector $x\in\mathbb{R}^n$, $\|x\|_2$ and $\|x\|_Q$ denote its 2-norm and $Q$-norm, i.e., $\|x\|^2_Q=x^TQx$, where $Q$ is a positive definite matrix. Finally, we use $I_n$ to denote the $n\times n$ identity matrix.
	
	\section{Problem Formulation and ILEMPC Design}\label{formulation}
	
	Consider a dynamic system
	\begin{equation}
	x(k+1)=f(x(k),u(k)),~x(0)=x^0,\label{sys}
	\end{equation}
	where $x\in\mathbb{X}\subset \mathbb{R}^n$ is the state and $u\in\mathbb{U}\subset \mathbb{R}^m$ is the control input, $\mathbb{X}$ and $\mathbb{U}$ are compact.
	
	Suppose that at the very beginning we have a feasible state and control sequence:
	\begin{eqnarray}
	x_0(0),~x_0(1),~x_0(2),\ldots;\nonumber\\
	u_0(0),~u_0(1),~u_0(2),\ldots,\nonumber
	\end{eqnarray}
	where $x_0(0)=x^0$, $x_0(i)\in\mathbb{X},~u_0(i)\in\mathbb{U},\forall i=0,1,2,\ldots.$
	
	The initial feasible state and control sequence could have an arbitrary pattern: converges to a steady state, or converges to a periodic trajectory or even is chaotic. For the simplicity of theoretical analysis, we assume that the length of the initial feasible state and control sequence is infinite. We use 
	\begin{eqnarray}
	x_j(0),~x_j(1),~x_j(2),\ldots;\nonumber\\
	u_j(0),~u_j(1),~u_j(2),\ldots,\nonumber
	\end{eqnarray}
	to collect the state and control sequences of the $j$-th iteration.
	
	We are interested in the following performance index, depending on which one is well defined:
	
	1)
	\begin{equation}
	\sum_{k=0}^{\infty}l(x(k),u(k));\nonumber
	\end{equation} 
	
	2)
	\begin{equation}
	\limsup_{T\to \infty}\frac{1}{T}\sum_{k=0}^{T-1}l(x(k),u(k))\nonumber
	\end{equation}
	and
	\begin{equation}
	\liminf_{T\to \infty}\frac{1}{T}\sum_{k=0}^{T-1}l(x(k),u(k)).\nonumber
	\end{equation}
	
	The stage cost $l(x,u)$ satisfies the following assumption:
	
	\begin{assumption}\label{continuity}
		$l(x,u)$ is continuous in $\mathbb{X}\times\mathbb{U}$.
	\end{assumption}
	
	Assume that for every iteration, the initial state of the system is $x^0$. We propose the following iterative learning economic MPC to optimize the performance index:
	For the $j$-th iteration, at time instant $k$, the following optimization problem is solved:
	
	\textbf{Problem~1}
	\begin{equation}
	\min_{u_j(k|k),\ldots,u_j(k+N-1|k)}\sum_{i=k}^{k+N-1}l(x_j(i|k),u_j(i|k))\nonumber
	\end{equation}
	subject to
	\begin{eqnarray}
	x_j(i+1|k)&=&f(x_j(i|k),u_j(i|k)),\nonumber\\
	x_j(i|k)&\in&\mathbb{X},\nonumber\\
	u_j(i|k)&\in&\mathbb{U},\nonumber\\
	x_j(k+N|k)&=&x_{j-1}(k+N),\label{terminal}\\
	x_j(k|k)&=&x_j(k),~i=k,\ldots,k+N-1.\nonumber
	\end{eqnarray}
	
	Denote the optimal solution and the corresponding state trajectory as
	\begin{eqnarray}
	u_j^*(k|k),~u_j^*(k+1|k),\ldots,~u_j^*(k+N-1|k),\nonumber
	\end{eqnarray}
	\begin{equation}
	x_j^*(k|k),~x_j^*(k+1|k),\ldots,~x_j^*(k+N-1|k),~x_j^*(k+N|k),\nonumber
	\end{equation}
	and the optimal value of the objective function as $V^*_j(x_j(k))=\sum_{i=k}^{k+N-1}l(x^*_j(i|k),u^*_j(i|k))$. The control input to be fed to the plant is $u_j(k)=u_j^*(k|k)$.
	
	\begin{theorem}\label{feasibility}
		Suppose that an initial feasible state and control sequence is available, then Problem 1 is feasible for any iteration $j$ and any time instant $k$.
	\end{theorem}
	\begin{proof}
		Suppose that after the $(j-1)$-th iteration, a feasible state and control sequence $x_{j-1}(k)$ and $u_{j-1}(k),~k\in\mathbb{N}$ is obtained. Then for the $j$-th iteration, at time instant $0$, the following state and control sequence is feasible, by the feasibility of $x_{j-1}(k)$ and $u_{j-1}(k),~k\in\mathbb{N}$:
		\begin{eqnarray}
		x_j(k|0)&=&x_{j-1}(k),~k=0,\ldots,N\nonumber\\
		u_j(k|0)&=&u_{j-1}(k),~k=0,\ldots,N-1.\nonumber
		\end{eqnarray}
		
		Suppose that at time instant $k$ of the $j$-th iteration, Problem 1 is feasible. By the terminal constraint \eqref{terminal}, we have $x_j^*(k+N|k)=x_{j-1}(k+N)$. Therefore, for time instant $k+1$, we can construct the following candidate solution:
		\begin{eqnarray}
		u_j^*(k+1|k),~u_j^*(k+2|k),\ldots,~u_j^*(k+N-1|k),~u_{j-1}(k+N)\nonumber
		\end{eqnarray}
		and the corresponding state trajectory
		\begin{equation}
		x_j^*(k+1|k),~x_j^*(k+2|k),\ldots,~x_j^*(k+N-1|k),~x_{j-1}(k+N),~x_{j-1}(k+N+1),\nonumber
		\end{equation}
		which are feasible by the feasibility of the $(j-1)$-th iteration. The theorem can be concluded by induction.
	\end{proof}
	
	\begin{remark}
		It might be desirable to use terminal inequality constraints to substitute terminal equality constraints for a larger set of feasible initial condition $x^0$ in standard non-iterative-learning MPC. However, in our problem formulation, such a set is determined by how an initial feasible trajectory is constructed, which is not related to the terminal constraint \eqref{terminal}. Therefore, equality terminal constraints in Problem 1 do not introduce strong conservativeness.
	\end{remark}
	
	\begin{remark}
		Compared with the algorithm proposed in \cite{Rosolia17}, our method has a simpler formulation and can be applied to more types of control tasks. In particular, we do not use a terminal cost in the objective function so we can handle the situations where $l(x(k),u(k))$ does not go to 0 as $k\to\infty$ such as imperfect tracking and economic optimization. Such a control task cannot be handled by the existing approach \cite{Rosolia17} directly since the terminal cost used there is the accumulative cost from time instant $k+N$ to $\infty$, which will be infinite. Furthermore, the terminal constraint \eqref{terminal} is a single equality constraint, resulting in a nonlinear programming used in standard MPC algorithm. On the other hand, the terminal constraint used in \cite{Rosolia17} is a set of discrete points, which collects data from all previous trajectories, leading to a mixed-integer nonlinear programming which is known to be NP hard in general. Finally, in terms of storage requirement, the algorithm proposed in this paper only uses the trajectory of the last iteration while the one in \cite{Rosolia17} requires trajectories of all previous iterations so the terminal constraint there will be more and more complex as the learning procedure continues and larger and larger storage space is required.
	\end{remark}

	\section{Analysis}\label{analysis}
	
	\subsection{Average performance analysis}
	
	We first consider the case when $x_0(k)$ does not converge to a steady state. This situation could happen in some chemical process where the performance index does not achieve optimal value at the steady state. In this case we are interested in the average performance of the state and control sequence.
	
	Denote $\bar{S}_j=\limsup_{T\to \infty}\frac{1}{T}\sum_{k=0}^{T-1}l(x_j(k),u_j(k))$ and $\underline{S}_j=\liminf_{T\to \infty}\frac{1}{T}\sum_{k=0}^{T-1}l(x_j(k),u_j(k))$. Then we have the following result:
	
	\begin{theorem}\label{average}
		If Assumption \ref{continuity} is satisfied, then $\bar{S}_{j+1}\le\underline{S}_{j},~j\in\mathbb{N}$.
	\end{theorem}
	
	\begin{proof}
		Suppose that at the $j$-th iteration and time instant $k$, we have the optimal solution of Problem 1:
		\begin{equation}
		u_j^*(k|k),~u_j^*(k+1|k),\ldots,~u_j^*(k+N-1|k)\nonumber
		\end{equation}
		and let the corresponding state trajectory be given by:
		\begin{equation}
		x_j^*(k|k),~x_j^*(k+1|k),\ldots,~x_j^*(k+N-1|k),~x_j^*(k+N|k).\nonumber
		\end{equation}
		Then we construct a feasible solution for time instant $k+1$ as in Theorem \ref{feasibility}:
		\begin{eqnarray}
		u_j^*(k+1|k),~u_j^*(k+2|k),\ldots,~u_j^*(k+N-1|k),~u_{j-1}(k+N)\nonumber
		\end{eqnarray}
		and the corresponding state trajectory
		\begin{equation}
		x_j^*(k+1|k),~x_j^*(k+2|k),\ldots,~x_j^*(k+N-1|k),~x_{j-1}(k+N),~x_{j-1}(k+N+1).\nonumber
		\end{equation}
		
		This trajectory may not be the optimal one of Problem 1 so
		\begin{eqnarray}
		V_j^*(x_j(k+1))\le\sum_{i=k+1}^{k+N-1}l(x_j^*(i|k),u_j^*(i|k))+l(x_{j-1}(k+N),u_{j-1}(k+N)).\nonumber
		\end{eqnarray} 
		
		Then one has that
		
		\begin{eqnarray}
		&&V_j^*(x_j(k+1))-V_j^*(x_j(k))\nonumber\\
		&\le&\sum_{i=k+1}^{k+N-1}l(x_j^*(i|k),u_j^*(i|k))+l(x_{j-1}(k+N),u_{j-1}(k+N))\nonumber\\
		&&-\sum_{i=k}^{k+N-1}l(x_j^*(i|k),u_j^*(i|k))\nonumber\\
		&=&-l(x_j(k),u_j(k))+l(x_{j-1}(k+N),u_{j-1}(k+N)).\nonumber
		\end{eqnarray}
		
		Taking average on both sides leads to that 
		\begin{eqnarray}
		&&\frac{1}{T}\sum_{k=0}^{T-1}[V_j^*(x_j(k+1))-V_j^*(x_j(k))]\nonumber\\
		&\le&-\frac{1}{T}\sum_{k=0}^{T-1}l(x_j(k),u_j(k))+\frac{1}{T}\sum_{k=0}^{T-1}l(x_{j-1}(k+N),u_{j-1}(k+N)).\nonumber
		\end{eqnarray}
		
		By Theorem \ref{feasibility}, the trajectory $x^*_j(i|k),~i=k,\ldots,k+N$ evolves in $\mathbb{X}$ for all $j,~k\in\mathbb{N}$. In view of Assumption \ref{continuity}, $\mathbb{X}$ and $\mathbb{U}$ are compact, hence $V_j^*(x_j(k))$ is bounded for any $j,~k\in\mathbb{N}$. Consequently, as $T\to\infty$, $\frac{1}{T}\sum_{k=0}^{T-1}[V_j^*(x_j(k+1))-V_j^*(x_j(k))]\to0$. On the other hand,
		\begin{eqnarray} &&\liminf_{T\to\infty}[-\frac{1}{T}\sum_{k=0}^{T-1}l(x_j(k),u_j(k))+\frac{1}{T}\sum_{k=0}^{T-1}l(x_{j-1}(k+N),u_{j-1}(k+N))]\nonumber\\
		&\le&-\limsup_{T\to\infty}\frac{1}{T}\sum_{k=0}^{T-1}l(x_j(k),u_j(k))+\liminf_{T\to\infty}\frac{1}{T}\sum_{k=0}^{T-1}l(x_{j-1}(k+N),u_{j-1}(k+N))\nonumber\\
		&=&-\bar{S}_j+\underline{S}_{j-1}.\nonumber 
		\end{eqnarray}
		Combining the above inequalities one can get that $\bar{S}_j\le\underline{S}_{j-1}$.
	\end{proof}
	
	\subsection{Convergence to a steady state}
	
	Now we consider the case when the initial feasible state and control sequence converges to a steady state. In this case, we assume that there exists $(x_s,u_s)$ satisfying that $x_s=f(x_s,u_s)$ and $\lim_{k\to\infty}x_0(k)=x_s$ and $\lim_{k\to\infty}u_0(k)=u_s$. Such a trajectory could be constructed by a known feasible feedback controller $u=\kappa(x)$. Then the proposed algorithm can be used to search for trajectories with better performance.
	
	We mainly discuss two properties of the proposed algorithm. The first one is that for the given initial feasible sequence, will the convergence property of the sequence be retained by running the iterative learning algorithm? The second one is that if the iteration converges, i.e., $\lim_{j\to\infty}x_j(k)$ and $\lim_{j\to\infty}u_j(k)$ exist for all $k\in\mathbb{N}$, what is this limit state and control sequence?
	
	For the convergence of the steady state, we make use of the following definition, which is standard in the economic MPC literature; see, e.g. \cite{Angeli12}.
	
	\begin{definition}
		System \eqref{sys} is dissipative with respect to a supply rate $s:~\mathbb{X}\times \mathbb{U}\to\mathbb{R}$ if there exists a continuous function $\lambda:\mathbb{X}\to\mathbb{R}$ such that:
		\begin{equation}
		\lambda(f(x,u))-\lambda(x)\le s(x,u)\nonumber
		\end{equation}
		for all $x\in\mathbb{X}$, $u\in\mathbb{U}$.
	\end{definition}
	
	This definition is slightly stronger than the one in \cite{Angeli12}. In \cite{Angeli12}, $\lambda(\cdot)$ is not required to be continuous.
	
	Define the rotated stage cost $L(x,u)=l(x,u)-\lambda(f(x,u))+\lambda(x)$. Then we introduce the auxiliary optimization problem:
	
	\textbf{Problem~2}
	\begin{equation}
	\min_{u_j(k|k),\ldots,u_j(k+N-1|k)}\sum_{i=k}^{k+N-1}L(x_j(i|k),u_j(i|k))\nonumber
	\end{equation}
	subject to
	\begin{eqnarray}
	x_j(i+1|k)&=&f(x_j(i|k),u_j(i|k)),\nonumber\\
	x_j(i|k)&\in&\mathbb{X},\nonumber\\
	u_j(i|k)&\in&\mathbb{U},\nonumber\\
	x_j(k+N|k)&=&x_{j-1}(k+N),\nonumber\\
	x_j(k|k)&=&x_j(k),~i=k,\ldots,k+N-1.\nonumber
	\end{eqnarray}
	Denote $\tilde{V}^*_j(x_j(k))=\sum_{i=k}^{k+N-1}L(x^*_j(i|k),u^*_j(i|k))$, where $u^*_j(i|k)$ and $x^*_j(i|k),~j=k,\ldots,k+N-1$ are respectively the optimal control and state sequences associated with $\tilde{V}_j^*(x_j(k))$.
	
	
\begin{assumption}\label{dissipative}
				System \eqref{sys} is dissipative with respect to the supply rate:
				\begin{equation}
				s(x,u)=l(x,u)-l(x_s,u_s).\nonumber
				\end{equation}
	\end{assumption}
	
	\begin{remark}
		By the definition of dissipative property, the following holds for any $x\in\mathbb{X}$, $u\in\mathbb{U}$
		\begin{equation}
		l(x_s,u_s)\le l(x,u)-\lambda(f(x,u))+\lambda(x).\nonumber
		\end{equation}
		
		Then,  
		\begin{equation}
		l(x_s,u_s)\le \min_{x,u}(l(x,u)-\lambda(f(x,u))+\lambda(x)).\nonumber
		\end{equation}
		
		On the other hand,
		\begin{equation}
		\min_{x,u}(l(x,u)-\lambda(f(x,u))+\lambda(x))\le l(x_s,u_s)-\lambda(f(x_s,u_s))+\lambda(x_s)=l(x_s,u_s).\nonumber
		\end{equation}
		
		Therefore, $(x_s,u_s)$ minimizes $l(x,u)-\lambda(f(x,u))+\lambda(x)$.
	\end{remark}
	
	Since $l(x_s,u_s)$ is the minimal value of  $L(x,u)$ and $l(x_s,u_s)$ is bounded, if $l(x_s,u_s)\neq 0$, we can redefine $l(x,u)$ by subtracting $l(x_s,u_s)$ from the original one. As a result, we can assume that $l(x_s,u_s)=0$ and $L(x,u)\ge 0$ in the sequel without loss of generality. 
	
	\begin{assumption}\label{unique}
		$(x_s,u_s)$ is the unique minimizer of $L(x,u)$.
	\end{assumption}
	
	\begin{lemma}
		For a given $x(k)$, Problem 1 and 2 have the same feasibility, i.e. if Problem 1 is feasible, then Problem 2 is also feasible and vice versa. If the problems are feasible, they have the same optimal solution(s).
	\end{lemma}
	
	\begin{proof}
		Firstly, note that Problem 1 and 2 have the same constraints. Therefore their feasibility is the same.
		
		Next, note that 
		\begin{eqnarray}
		&&\sum_{i=k}^{k+N-1}L(x_j(i|k),u_j(i|k))\nonumber\\
		&=&\sum_{i=k}^{k+N-1}l(x_j(i|k),u_j(i|k))-\lambda(f(x_j(i|k),u_j(i|k)))+\lambda(x_j(i|k))\nonumber\\
		&=&\sum_{i=k}^{k+N-1}l(x_j(i|k),u_j(i|k))+\lambda(x_j(k|k))-\lambda(x_j(k+N|k))\nonumber\\
		&=&\sum_{i=k}^{k+N-1}l(x_j(i|k),u_j(i|k))+\lambda(x_j(k))-\lambda(x_{j-1}(k+N))\nonumber
		\end{eqnarray}
		and $\lambda(x_j(k))-\lambda(x_{j-1}(k+N))$ is a constant, since at time instant $k$ of the $j$-th iteration, $x_j(k)$ and $x_{j-1}(k+N)$ are both known. Consequently, the objective functions of Problem 1 and 2 only differ by a constant, which implies that they have the same optimal solution(s).
	\end{proof}
	
	%
	\begin{lemma}\label{mean}
		For a non-negative sequence $\{a_n\},~n\in\mathbb{N}$, if $\lim_{N\to\infty}\frac{\sum_{i=0}^{N-1}a_i}{N}=0$, then there exist an infinite subsequence $\{a_{k_n}\}$ such that $\lim_{n\to\infty}a_{k_n}=0$.
	\end{lemma}
	
	\begin{proof}
		We prove it by contradiction. Suppose that such an infinite subsequence does not exist. Then there exists an integer $N$ and positive constant $\epsilon$ such that $\forall n\ge N$, $a_n\ge\epsilon$. However, in this case, $\lim_{N\to\infty}\frac{\sum_{i=0}^{N-1}a_i}{N}\ge\epsilon>0$, which contradicts the fact that $\lim_{N\to\infty}\frac{\sum_{i=0}^{N-1}a_i}{N}=0$.
	\end{proof}
	
	\begin{lemma}\label{convergence}
				If Assumption \ref{continuity}, \ref{dissipative} and \ref{unique} are satisfied,
				$\lim_{k\to\infty}(x_j(k),u_j(k))=(x_s,u_s)$ and $(x_s,u_s)$ is the unique minimizer of $L(x,u)$, then there exists a subsequence of $x_{j+1}(k)$ and $u_{j+1}(k)$ such that 
				\begin{equation}
				\lim_{n\to\infty}(x_{j+1}(k_n),u_{j+1}(k_n))=(x_s,u_s).\nonumber
				\end{equation}
	\end{lemma}
	
	\begin{proof}
		Since Problem 1 and 2 are equivalent, we study the proposed iterative learning algorithm induced by Problem 2. By the continuity of $L(x,u)$, $L(x,u)$ is upper bounded in $\mathbb{X}\times\mathbb{U}$. 
		
		Following the same argument in the proof of Theorem \ref{average}, one can write that
		\begin{equation}
		\tilde{V}^*_{j+1}(x_{j+1}(k+1))-\tilde{V}^*_{j+1}(x_{j+1}(k))\le L(x_{j}(k+N),u_{j}(k+N))-L(x_{j+1}(k),u_{j+1}(k)).\nonumber
		\end{equation}
		Taking average on both sides leads to that 
		\begin{equation}
		\frac{1}{T}\sum_{k=0}^{T-1}L(x_{j+1}(k),u_{j+1}(k))\le \frac{1}{T}\sum_{k=0}^{T-1}L(x_{j}(k+N),u_{j}(k+N))+\frac{\tilde{V}^*_{j+1}(x_{j+1}(0))-\tilde{V}^*_{j+1}(x_{j+1}(T))}{T}.\nonumber
		\end{equation}
		By recursive feasibility of Problem 2 and the continuity of $L(x,u)$, we know that both ${V}^*_{j+1}(x_{j+1}(0))$ and $\tilde{V}^*_{j+1}(x_{j+1}(T))$ are bounded.
		
		Since $\lim_{k\to\infty}x_{j}(k)=x_s$ and $\lim_{k\to\infty}u_{j}(k)=u_s$, $L(x_{j}(k),u_{j}(k))\to 0$ as $k\to\infty$. Therefore, we have $\lim_{T\to\infty}\frac{1}{T}\sum_{k=0}^{T-1}L(x_{j}(k+N),u_{j}(k+N))=0$ and $\lim_{T\to\infty}\frac{1}{T}\sum_{k=0}^{T-1}L(x_{j+1}(k),u_{j+1}(k))=0$. Then by Lemma \ref{mean} there exists a subsequence $L(x_{j+1}(k_n),u_{j+1}(k_n))\to 0$ as $n\to\infty$. By the uniqueness of $(x_s,u_s)$, one has that $\lim_{n\to\infty}x_{j+1}(k_n)=x_s$ and $\lim_{n\to\infty}u_{j+1}(k_n)=u_s$.
			\end{proof}


	Then we prove that the closed-loop performance of each iteration is not worse than that of the previous one. To this end, we assume that $\lim_{T\to\infty}\sum_{k=0}^{T}L(x_0(k),u_0(k))<\infty$. Note that $\lim_{T\to\infty}\sum_{k=0}^{T}L(x_0(k),u_0(k))=\lim_{T\to\infty}\sum_{k=0}^{T}l(x_0(k),u_0(k))+\lambda(x(0))-\lambda(x_s)$. Thus, $\lim_{T\to\infty}\sum_{k=0}^{T}l(x_0(k),u_0(k))<\infty$. Also note that if $\lim_{T\to\infty}\sum_{k=0}^{T}L(x_0(k),u_0(k))=\infty$ but $\limsup_{T\to\infty}\frac{1}{T}\sum_{k=0}^{T-1}L(x_0(k),u_0(k))$ and $\liminf_{T\to\infty}\frac{1}{T}\sum_{k=0}^{T-1}L(x_0(k),u_0(k))$ exist, then Theorem \ref{average} can be applied directly. Denote $J_j=\lim_{T\to\infty}\sum_{k=0}^{T}l(x_j(k),u_j(k))$ and\\ $\tilde{J}_j=\lim_{T\to\infty}\sum_{k=0}^{T}L(x_j(k),u_j(k))$.
	
	\begin{lemma}\label{convergeperformance}
		If Assumptions \ref{continuity} and \ref{dissipative} are satisfied, $\lim_{k\to\infty}(x_j(k),u_j(k))=(x_s,u_s)$ and there exists a subsequence of $x_{j+1}(k)$ and $u_{j+1}(k)$ such that 
		\begin{equation}
		\lim_{n\to\infty}(x_{j+1}(k_n),u_{j+1}(k_n))=(x_s,u_s),\nonumber
		\end{equation}
		and $J_j<\infty$, then $J_{j+1}\le J_j$.
	\end{lemma}
	
	\begin{proof}
		Similar to the proof of Lemma \ref{convergence}, we have
		\begin{equation}
		\tilde{V}^*_{j+1}(x_{j+1}(k+1))-\tilde{V}^*_{j+1}(x_{j+1}(k))\le L(x_{j}(k+N),u_{j}(k+N))-L(x_{j+1}(k),u_{j+1}(k)).\nonumber
		\end{equation}
		Taking summation on both sides leads to that
		 
				\begin{equation}
				\sum_{i=0}^{k_n-1}L(x_{j+1}(i),u_{j+1}(i))\le \sum_{i=0}^{k_n-1}L(x_{j}(i+N),u_{j}(i+N))+\tilde{V}^*_{j+1}(x_{j+1}(0))-\tilde{V}^*_{j+1}(x_{j+1}(k_n)).\label{limiteperformance}
				\end{equation}
			
			Since $\lim_{n\to\infty}x_{j+1}(k_n)=x_s$ and $\lim_{k\to\infty}x_j(k)=x_s$, $\lim_{n\to\infty}\tilde{V}^*_{j+1}(x_{j+1}(k_n))=0$. By letting $n\to\infty$, \eqref{limiteperformance} becomes
		\begin{equation}
		\tilde{J}_{j+1}\le \tilde{J}_{j}-\sum_{k=0}^{N-1}L(x_{j}(k),u_{j}(k))+\tilde{V}^*_{j+1}(x_{j+1}(0)).\label{limit}
		\end{equation}
		Note that the right hand side of \eqref{limit} is finite so this inequality is well defined. Consider Problem 2 at $k=0$ and iteration $j+1$. By recursive feasibility of the proposed algorithm, one can observe that $x_{j}(k),~k=0,\ldots,N$ and $u_{j}(k),~k=0,\ldots,N-1$ are feasible state and control sequences for Problem 2 at $k=0$ and iteration $j+1$. Therefore, the cost $\sum_{k=0}^{N-1}L(x_{j}(k),u_{j}(k))$ will not be smaller than the optimal one, which is $\tilde{V}^*_{j+1}(x_{j+1}(0))$. Combining this fact with \eqref{limit}, one obtains that $\tilde{J}_{j+1}\le \tilde{J}_j,~j\in\mathbb{N}$. The proof is completed by noticing that $\tilde{J}_j=J_j+\lambda(x(0))-\lambda(x_s)$.
	\end{proof}
	
			The first main result of this section is summarized in the following theorem.
			
			\begin{theorem}\label{converge}
				If Assumption \ref{continuity}, \ref{dissipative} and \ref{unique} are satisfied, $\lim_{k\to\infty}(x_0(k),u_0(k))=(x_s,u_s)$, and $\lim_{T\to\infty}\sum_{k=0}^{T}L(x_0(k),u_0(k))<\infty$, then
				\begin{equation}
				J_{j+1}\le J_j\nonumber
				\end{equation}
				and 
				\begin{equation}
				\lim_{k\to\infty}(x_{j}(k),u_{j}(k))=(x_s,u_s),~\forall j\in\mathbb{N}.\nonumber
				\end{equation}
			\end{theorem}
			
			\begin{proof}
			It is easy to see that if $\tilde{J}_j<\infty$, then we have $\lim_{k\to\infty}L(x_j(k),u_j(k))=0$ and $\lim_{k\to\infty}(x_j(k),u_j(k))=(x_s,u_s)$ by the uniqueness of $(x_s,u_s)$. Applying Lemma \ref{convergence} to $j=0$, we have that there exists a subsequence of $x_1(k)$ and $u_1(k)$ such that 
			\begin{equation}
			\lim_{n\to\infty}(x_{1}(k_n),u_{1}(k_n))=(x_s,u_s).\nonumber
			\end{equation}
			Then we apply Lemma \ref{convergeperformance} to obtain that $J_{1}\le J_0<\infty$, which implies that $\lim_{k\to\infty}(x_{1}(k),u_{1}(k))=(x_s,u_s).$ The theorem can be concluded by induction.
	\end{proof}
	
	\begin{remark}
		In \cite{Angeli12}, to ensure the stability of the steady state, strictly dissipative assumption, which requires a positive definite function $\rho(x)$ with respect to $(x_s,u_s)$ such that $\lambda(f(x,u))-\lambda(x)\le -\rho(x)+s(x,u)$, is needed. In this paper, we only require a weaker dissipative assumption, if a convergent initial feasible trajectory is available.
	\end{remark}
	
	\begin{definition}\label{recedingopt}
		Given state and control sequences $x(k),~u(k),~k\in\mathbb{N}$. They are $N$-receding-horizon optimal for system \eqref{sys} if $u(i),~i=k,\ldots,k+N-1$ and $x(i),~i=k,\ldots,k+N$ are the optimal solution and the corresponding state sequence of the following problem $\forall k\in\mathbb{N}$:
		
		\textbf{Problem~3}
		\begin{equation}
		\min_{u(k|k),\ldots,u(k+N-1|k)}\sum_{i=k}^{k+N-1}l(x(i|k),u(i|k))\nonumber
		\end{equation}
		subject to
		\begin{eqnarray}
		x(i+1|k)&=&f(x(i|k),u(i|k)),\nonumber\\
		x(i|k)&\in&\mathbb{X},\nonumber\\
		u(i|k)&\in&\mathbb{U},\nonumber\\
		x(k|k)&=&x(k),\nonumber\\
		x(k+N|k)&=&x(k+N).\nonumber
		\end{eqnarray}
	\end{definition}
	
	\begin{corollary}
		$N$-receding-horizon optimality implies $(N-1)$-receding-horizon optimality.
\end{corollary}
		\begin{proof}
			Suppose that $x(k),~u(k),~k\in\mathbb{N}$ are $N$-receding-horizon optimal but not $(N-1)$-receding-horizon optimal. Then there exists some $k\in\mathbb{N}$ such that there exist another feasible trajectory $\tilde{x}(i),~i=k,\ldots,k+N-1$ and associated control sequence $\tilde{u}(i),~i=k,\ldots,k+N-2$ such that $\sum_{i=k}^{k+N-2}l(\tilde{x}(i),\tilde{u}(i))<\sum_{i=k}^{k+N-2}l(x(i),u(i))$ and $\tilde{x}(k)=x(k)$ and $\tilde{x}(k+N-1)=x(k+N-1)$. Now we consider $\hat{x}(i),~i=k,\ldots,k+N$, which is constructed as $\hat{x}(i)=\tilde{x}(i),~i=k,\ldots,k+N-1$ and $\hat{x}(k+N)=x(k+N)$, and $\hat{u}(i),~i=k,\ldots,k+N-1$, which is constructed as $\hat{u}(i)=\tilde{u}(i),~i=k,\ldots,k+N-2$ and $\hat{u}(k+N-1)=u(k+N-1)$. Then one has that $\sum_{i=k}^{k+N-1}l(\hat{x}(i),\hat{u}(i))=\sum_{i=k}^{k+N-2}l(\tilde{x}(i),\tilde{u}(i))+l(x(k+N-1),u(k+N-1))<\sum_{i=k}^{k+N-2}l(x(i),u(i))+l(x(k+N-1),u(k+N-1))=\sum_{i=k}^{k+N-1}l(x(i),u(i))$, which contradicts the fact that $x(k),~u(k),~k\in\mathbb{N}$ are $N$-receding-horizon optimal.
		\end{proof}


			\begin{assumption}\label{uniquesolution}
				The optimal solution of Problem 1 is unique.
			\end{assumption}
			
			Since Problem 1 and 2 are equivalent, the optimal solution of Problem 2 is also unique if Assumption \ref{uniquesolution} is satisfied.
			
			\begin{lemma}\label{eqimplyeq}
				If Assumption \ref{continuity}, \ref{dissipative} and \ref{uniquesolution} are satisfied and there exists some positive integer $N_0$ such that $\tilde{J}_{N_0+1}=\tilde{J}_{N_0}$, then $x_{j+1}(k)=x_j(k)$ and $u_{j+1}(k)=u_j(k),~\forall k\in\mathbb{N},~j\ge N_0$. 
				\end{lemma}
				
				\begin{proof}
					Denote $\tilde{J}_j(k)=\lim_{T\to\infty}\sum_{i=k}^{T+k-1}L(x_j(i),u_j(i))$. Suppose that $\tilde{J}_{N_0}(k)=\tilde{J}_{N_0+1}(k)$ and $x_{N_0}(k)=x_{N_0+1}(k)$. Similar to the proof of Lemma \ref{convergeperformance}, one can write that 
					\begin{equation}
					\tilde{J}_{N_0+1}(k)\le \tilde{J}_{N_0}(k)-\sum_{i=k}^{k+N-1}L(x_{N_0}(i),u_{N_0}(i))+\tilde{V}^*_{N_0+1}(x_{N_0+1}(k)).\nonumber
					\end{equation}
					
					Since $\tilde{J}_{N_0+1}(k)=\tilde{J}_{N_0}(k)$, the above implies that $\sum_{i=k}^{k+N-1}L(x_{N_0}(i),u_{N_0}(i))\le\tilde{V}^*_{N_0+1}(x_{N_0+1}(k))$. Since $x_{N_0}(k)=x_{N_0+1}(k)$, $x_{N_0}(i),~i=k,\ldots,k+N$ and $u_{N_0}(i),~i=k,\ldots,k+N-1$ are feasible for Problem 2 at time instant $k$ and the $(N_0+1)$-th iteration. As a result, $\sum_{i=k}^{k+N-1}L(x_{N_0}(i),u_{N_0}(i))\ge\tilde{V}^*_{N_0+1}(x_{N_0+1}(k))$. Consequently, we have $\sum_{i=k}^{k+N-1}L(x_{N_0}(i),u_{N_0}(i))=\tilde{V}^*_{N_0+1}(x_{N_0+1}(k))$ and $x_{N_0}(k+1)=x_{N_0+1}(k+1)$ and $u_{N_0}(k)=u_{N_0+1}(k)$ by Assumption \ref{uniquesolution}. Furthermore, $\tilde{J}_{N_0+1}(k+1)=\tilde{J}_{N_0}(k+1)$ since $\tilde{J}_{N_0+1}(k+1)=\tilde{J}_{N_0+1}(k)-L(x_{N_0+1}(k),u_{N_0+1}(k))$, $\tilde{J}_{N_0}(k+1)=\tilde{J}_{N_0}(k)-L(x_{N_0}(k),u_{N_0}(k))$, $\tilde{J}_{N_0+1}(k)=\tilde{J}_{N_0}(k)$ and $L(x_{N_0+1}(k),u_{N_0+1}(k))=L(x_{N_0}(k),u_{N_0}(k))$. Note that $\tilde{J}_{N_0+1}(0)=\tilde{J}_{N_0}(0)$ and $x_{N_0+1}(0)=x_{N_0}(0)$. By induction one has that $x_{N_0+1}(k)=x_{N_0}(k)$ and $u_{N_0+1}(k)=u_{N_0}(k),~\forall k\in\mathbb{N}.$ Now suppose that $x_{j+1}(k)=x_{j}(k),~\forall k\in\mathbb{N}$ and $x_{j+2}(k_0)=x_{j+1}(k_0)$ for some $k_0\in\mathbb{N}$. Then Problem 2 formulated at time instant $k_0$ of the $(j+2)$-th iteration is the same as the one formulated at time instant $k_0$ of the $(j+1)$-th iteration. Therefore the solutions of both problem are identical and it results that $x_{j+2}(k_0+1)=x_{j+1}(k_0+1)$. The proof is completed by induction and the fact that $x_{j+1}(0)=x_{j}(0),~\forall j\in\mathbb{N}$.
				\end{proof}
			
				Lemma \ref{eqimplyeq} implies that if the performance index converges in finite steps， then the state and control sequences converge as well.
	
Suppose that the limits $\lim_{j\to\infty}x_j(k)$ and $\lim_{j\to\infty}u_j(k)$ exist and we denote $\lim_{j\to\infty}x_j(k)$ and $\lim_{j\to\infty}u_j(k)$ as $x_{\infty}(k)$ and $u_{\infty}(k)$ respectively and $\tilde{V}_{\infty}\triangleq\sum_{k=0}^{\infty}L(x_{\infty}(k),u_{\infty}(k))$, which is finite.
	
	\begin{theorem}\label{convergesuboptimal}
		If Assumption \ref{continuity}, \ref{dissipative} and \ref{uniquesolution} are satisfied, then $x_{\infty}(k)$ and $u_{\infty}(k)$ are $N$-receding-horizon optimal for system \eqref{sys}.
	\end{theorem}
	
	\begin{proof}
		Denote the optimal costs of Problem 1 and Problem 2 at time instant $k$ and under initial state $x_{\infty}(k)$ as $V_{\infty,N}(k)$ and $\tilde{V}_{\infty,N}(k)$, respectively. Then similar to the proof of Theorem \ref{converge}, one has $\tilde{V}_{\infty,N}(k+1)-\tilde{V}_{\infty,N}(k)\le L(x_{\infty}(k+N),u_{\infty}(k+N))-L(x_{\infty}(k),u_{\infty}(k))$. Taking summation on both sides from $k$ to $k+T-1$ leads to that
		\begin{equation}
		\tilde{V}_{\infty,N}(k+T)-\tilde{V}_{\infty,N}(k)\le \sum_{i=k}^{k+T-1}L(x_{\infty}(i+N),u_{\infty}(i+N))-\sum_{i=k}^{k+T-1}L(x_{\infty}(i),u_{\infty}(i)).\nonumber
		\end{equation}
		The above implies that
		\begin{equation}
		\tilde{V}_{\infty,N}(k)+\sum_{i=k}^{k+T-1}L(x_{\infty}(i+N),u_{\infty}(i+N))\ge\sum_{i=k}^{k+T-1}L(x_{\infty}(i),u_{\infty}(i)).\label{ineq1}
		\end{equation}
		
		Letting $T\to\infty$, we obtain that
		\begin{equation} \lim_{T\to\infty}\sum_{i=k}^{k+T-1}L(x_{\infty}(i+N),u_{\infty}(i+N))=\tilde{V}_{\infty}-\sum_{i=0}^{k+N-1}L(x_{\infty}(i),u_{\infty}(i))\nonumber
		\end{equation} 
		and
		\begin{equation} \lim_{T\to\infty}\sum_{i=k}^{k+T-1}L(x_{\infty}(i),u_{\infty}(i))=\tilde{V}_{\infty}-\sum_{i=0}^{k-1}L(x_{\infty}(i),u_{\infty}(i)).\nonumber\end{equation} 
		Then \eqref{ineq1} becomes
		\begin{equation}
		\tilde{V}_{\infty,N}(k)+\tilde{V}_{\infty}-\sum_{i=0}^{k+N-1}L(x_{\infty}(i),u_{\infty}(i))\ge \tilde{V}_{\infty}-\sum_{i=0}^{k-1}L(x_{\infty}(i),u_{\infty}(i)),\nonumber
		\end{equation}
		which is
		\begin{equation}
		\tilde{V}_{\infty,N}(k)\ge\sum_{i=k}^{k+N-1}L(x_{\infty}(k),u_{\infty}(k)).\label{VL}
		\end{equation}
		
		On the other hand, $u_{\infty}(i),~i=k,\ldots,l+N-1$ is a feasible solution of Problem 2. Therefore, $\tilde{V}_{\infty,N}(k)\le\sum_{i=k}^{k+N-1}L(x_{\infty}(k),u_{\infty}(k))$. So we can conclude that $\tilde{V}_{\infty,N}(k)=\sum_{i=k}^{k+N-1}L(x_{\infty}(k),u_{\infty}(k))$ and the proof is completed by the uniqueness of the optimal solution of Problem 1 and the fact that $\tilde{V}_{\infty,N}(k)-V_{\infty,N}(k)=\sum_{i=k}^{k+N-1}L(x_{\infty}(k),u_{\infty}(k))-\sum_{i=k}^{k+N-1}l(x_{\infty}(k),u_{\infty}(k))=\lambda(x_{\infty}(k))-\lambda(x_{\infty}(k+N))$.
	\end{proof}
	
	\begin{remark}
		Consider the extreme case when $N=1$, time instant $k=0$ and iteration index $j=1$. The initial condition of Problem 1 is fixed as $x_1(0)=x^0$ and the terminal condition is fixed as $x_1(1)=x_0(1)$. Then by solving Problem 1, the controller recovers $u_0(0)$ and it will be applied to the plant. By repeating this procedure, it is not hard to see that, when $N=1$, the closed-loop system always recovers the initial feasible trajectory at each iteration. When $N>1$, there will be free decision variables for the controller to improve performance. As a result, $N$ represents the capacity of innovation the controller has during the learning process. $N=1$ means that the controller can only copy from previous iteration and only when $N>1$, the controller can has the capacity to explore new information to achieve better performance in the learning process.
	\end{remark}
	
	\section{Extension To Cases With Average Constraint}\label{avecon}
	
	In real applications, point-wise in time hard constraints may make the closed-loop performance too conservative. For example, in the control of heating, ventilation and air-conditioning (HVAC) system, one may define a thermal comfort constraint as $[\underline{T},\bar{T}]$, where $\underline{T}$ and $\bar{T}$ represent the lower and upper bounds of indoor temperature. To keep indoor temperature $T(k)$ inside the given bounds for all time is an energy consuming strategy. A lot of literature have discussed the situation when the hard constraint is relaxed as a probabilistic one: $P(T(k)\notin[\underline{T},\bar{T}])\le\epsilon,~\forall k\in\mathbb{N}$ and used stochastic MPC to handle such a constraint \cite{YMASMPC,Korda14,Oldewurtel10}. On the other hand, we may interpret the violation probability $\epsilon$ as the frequency of the event $T(k)\notin[\underline{T},\bar{T}]$ that happens. Then one may define an indicator function $y(k)=f(T(k))$, where $f(x)$ is defined as 
	\begin{equation}
	f(x)=\begin{cases}0,~x\in[\underline{T},\bar{T}],\nonumber\\
	1,~\text{otherwise},\nonumber
	\end{cases}
	\end{equation}
	and require that $\lim_{N\to\infty}\frac{\sum_{i=0}^{N-1}y(i)}{N}\le\epsilon$, which imitates the probabilistic constraint to some extent. 
	
	Consider a bounded sequence $v(k),~k\in\mathbb{N}$. Similar to \cite{Angeli12}, we define the set of asymptotic averages:
	\begin{equation}
	Av[v]=\{\bar{v}|\exists k_n\to\infty:~\lim_{n\to\infty}\frac{\sum_{i=0}^{k_n}v(i)}{k_n+1}=\bar{v}\}.\nonumber
	\end{equation}
	
	Let $\mathbb{Y}\subset\mathbb{R}^p$ be a compact convex set and $y_j$ an auxiliary output variable defined as $y_j(k)=h(x_j(k),u_j(k))$, where $h:\mathbb{X}\times\mathbb{U}\to\mathbb{R}^p$ is continuous. We will discuss on how to ensure that the average constraint $Av[y_j]\subset\mathbb{Y}$ can be satisfied by state and control sequences for any iteration $j\in\mathbb{N}$. 
	
	First, we assume that for the initial feasible state and control sequences $x_0(k)$ and $u_0(k)$, $Av(y_0)\in\mathbb{Y},~k\in\mathbb{N}$ holds. Then we introduce the following optimization problem for the $j$-th iteration at time instant $k$
	
	\textbf{Problem~4}
	\begin{equation}
	\min_{u_j(k|k),\ldots,u_j(k+N-1|k)}\sum_{i=k}^{k+N-1}l(x_j(i|k),u_j(i|k))\nonumber
	\end{equation}
	subject to
	\begin{eqnarray}
	x_j(i+1|k)&=&f(x_j(i|k),u_j(i|k)),\nonumber\\
	x_j(i|k)&\in&\mathbb{X},\nonumber\\
	u_j(i|k)&\in&\mathbb{U},\nonumber\\
	x_j(k+N|k)&=&x_{j-1}(k+N),\label{terminal2}\\
	\sum_{i=k}^{k+N-1}h(x_j(i|k),u_j(i|k))&\in&\mathbb{Y}_{j,k}\label{avc}\\
	x_j(k|k)&=&x_j(k),~i=k,\ldots,k+N-1.\nonumber
	\end{eqnarray}
	The time-varying constraint $\mathbb{Y}_{j,k}$ is used to ensure that the average constraint can be satisfied. $\mathbb{Y}_{j,k}$ is constructed as $\mathbb{Y}_{j,k+1}=\mathbb{Y}_{j,k}\ominus h(x_j(k),u_j(k))\oplus h(x_{j-1}(k+N),u_{j-1}(k+N))$, with $\mathbb{Y}_{j,0}=\mathbb{Y}_{0}\oplus\sum_{i=0}^{N-1}h(x_{j-1}(i),u_{j-1}(i))$ and $\mathbb{Y}_0\subset\mathbb{R}^p$ being an arbitrary compact convex set containing the origin. $\ominus$ and $\oplus$ denote standard set subtraction and addition respectively.
	
	\begin{theorem}
		Under the initial feasible state and control sequences $x_0(k)$ and $u_0(k),~k\in\mathbb{N}$, Problem 4 is feasible for any time instant $k$ of the $j$-th iteration. Moreover, $Av[y_j]$ holds for all $j\in\mathbb{N}$.
	\end{theorem}
	
	\begin{proof}
		Suppose that after the $(j-1)$-th iteration, feasible state and control sequences $x_{j-1}(k)$ and $u_{j-1}(k)$, $k\in\mathbb{N}$ are obtained. Then for the $j$-th iteration, at time instant $0$, the following state and control sequence is feasible:
		\begin{eqnarray}
		x_j(i|0)&=&x_{j-1}(i),~i=0,\ldots,N,\nonumber\\
		u_j(i|0)&=&u_{j-1}(i),~i=0,\ldots,N-1,\nonumber
		\end{eqnarray}
		since $x_j(N|0)=x_{j-1}(N)$ and $\sum_{i=0}^{N-1}h(x_{j-1}(i),u_{j-1}(i))\in\mathbb{Y}_{j,0}$.
		
		Suppose that at time instant $k$ of the $j$-th iteration, Problem 4 is feasible. By the terminal constraint \eqref{terminal2}, we have $x_j^*(k+N|k)=x_{j-1}(k+N)$. Therefore, for time instant $k+1$, we can construct the following candidate solution:
		\begin{eqnarray}
		u_j^*(k+1|k),~u_j^*(k+2|k),\ldots,~u_j^*(k+N-1|k),~u_{j-1}(k+N)\nonumber
		\end{eqnarray}
		and let the corresponding state trajectory be given by
		\begin{equation}
		x_j^*(k+1|k),~x_j^*(k+2|k),\ldots,~x_j^*(k+N-1|k),~x_{j-1}(k+N),~x_{j-1}(k+N+1).\nonumber
		\end{equation}
		To prove that this candidate solution is feasible, we only need to show that constraint \eqref{avc} is satisfied since other constraints are satisfied by Theorem \ref{feasibility}. Note that $\sum_{i=k}^{k+N-1}h(x_j^*(i|k),u_j^*(i|k))\in\mathbb{Y}_{j,k}$. Feasibility directly follows from that
		\begin{eqnarray}
		&&\sum_{i=k+1}^{k+N-1}h(x_j^*(i|k),u_j^*(i|k))+h(x_{j-1}(k+N),u_{j-1}(k+N))\nonumber\\
		&\in&\mathbb{Y}_{j,k}\ominus h(x_j(k),u_j(k))\oplus h(x_{j-1}(k+N),u_{j-1}(k+N))\nonumber\\
		&=&\mathbb{Y}_{j,k+1}\nonumber
		\end{eqnarray}
		
		To show that the average constraint is also satisfied, we first rewrite $\mathbb{Y}_{j,k}$ as 
		\begin{eqnarray}
		\mathbb{Y}_{j,k}&=&\mathbb{Y}_{j,0}\oplus\sum_{i=0}^{k-1}h(x_{j-1}(i+N),u_{j-1}(i+N))\ominus\sum_{i=0}^{k-1}h(x_j(i),u_j(i))\nonumber\\
		&=&\mathbb{Y}_{0}\oplus\sum_{i=0}^{k+N-1}h(x_{j-1}(i),u_{j-1}(i))\ominus\sum_{i=0}^{k-1}h(x_j(i),u_j(i)).\label{rewriteavc}
		\end{eqnarray}
		Combining \eqref{rewriteavc} with \eqref{avc} implies that
		\begin{eqnarray}
		\sum_{i=0}^{k-1}h(x_j(i),u_j(i))+\sum_{i=k}^{k+N-1}h(x_j(i|k),u_j(i|k))\in\mathbb{Y}_{0}\oplus\sum_{i=0}^{k+N-1}h(x_{j-1}(i),u_{j-1}(i)).\nonumber
		\end{eqnarray}
		Note that $\sum_{i=k}^{k+N-1}h(x_j(i|k),u_j(i|k))$ is bounded due to the compactness of $\mathbb{X}$ and $\mathbb{U}$ and the continuity of $h(\cdot,\cdot)$. Then by letting $k$ goes to infinity along any subsequence $k_n$ such that $\lim_{n\to\infty}\frac{\sum_{i=0}^{i}y_j(k_n)}{k_n+1}$ exists, we have
		\begin{eqnarray}
		\lim_{n\to\infty}\frac{\sum_{i=0}^{k_n}y_j(k_n)}{k_n+1}\in\lim_{n\to\infty}\frac{\mathbb{Y}_{0}\oplus\sum_{i=0}^{k_n+N-1}h(x_{j-1}(i),u_{j-1}(i))}{k_n}\in\mathbb{Y},\nonumber
		\end{eqnarray}
		by the feasibility of $x_{j-1}(k)$ and $u_{j-1}(k),~k\in\mathbb{N}$. Then the proof is concluded by induction and the feasibility of $x_{0}(k)$ and $u_{0}(k),~k\in\mathbb{N}$.
	\end{proof}

	\section{Numerical Examples}\label{example}
	
	All the following examples are implemented with ICLOCS \cite{ICLOCS} and solved by IPOPT \cite{ipopt}.
	
	\subsection{Constrained regulator}
	\subsubsection{Linear case}
	
	We first test the proposed iterative learning MPC on the same example as in \cite{Rosolia17}. The system model is given by
	\begin{equation}
	x(k+1)=Ax(k)+Bu(k),\nonumber
	\end{equation}
	with $A=\begin{pmatrix}1&1\\0&1\end{pmatrix}$, $B=\begin{pmatrix}0\\1\end{pmatrix}$ and $x(0)=\begin{pmatrix}-3.95\\-0.05\end{pmatrix}$. The constraints for the system are 
	\begin{equation}
	\begin{pmatrix}-4\\-4\end{pmatrix}\le x(k)\le\begin{pmatrix}4\\4\end{pmatrix},~k\in\mathbb{N}\nonumber
	\end{equation}
	and
	\begin{equation}
	-1\le u(k)\le 1,~k.\in\mathbb{N}\nonumber
	\end{equation}
	
	The performance index to be minimized is $\sum_{k=0}^{\infty}\|x(k)\|_2^2+\|u(k)\|_2^2$. The initial feasible state and control sequence is generated by using an open-loop controller to drive the state to a small neighborhood of the origin and then using a stabilizing linear feedback controller. The prediction horizon is also chosen as 4 as in \cite{Rosolia17}. In Fig.~\ref{Fig:1} we present the evolution of the performance along iterative learning. In Fig.~\ref{Fig:2}, the state trajectories of iterations are shown. It can be observed that it converges after 5 iterations. After 15 iterations, the value of performance index is 49.9163600440, which is the same as the exact optimal one in \cite{Rosolia17} within 10 digits after the decimal point. Though the resulted closed-loop trajectory is only $4$-receding-horizon optimal by Theorem \ref{convergesuboptimal}, the performance compared with the optimal trajectory is almost the same. Within the digit limit of the used numerical solver in this paper, the best closed-loop performance can be achieved is 49.916360043958505 when prediction horizon $N\ge 5$. Note that in this example, at each time instant only a standard quadratic programming is solved while in \cite{Rosolia17} the controller needs to solve a mixed-integer programming, which is significantly more complex than the standard quadratic programming.
	
	\begin{figure}
		\centering
		\includegraphics[scale=0.52]{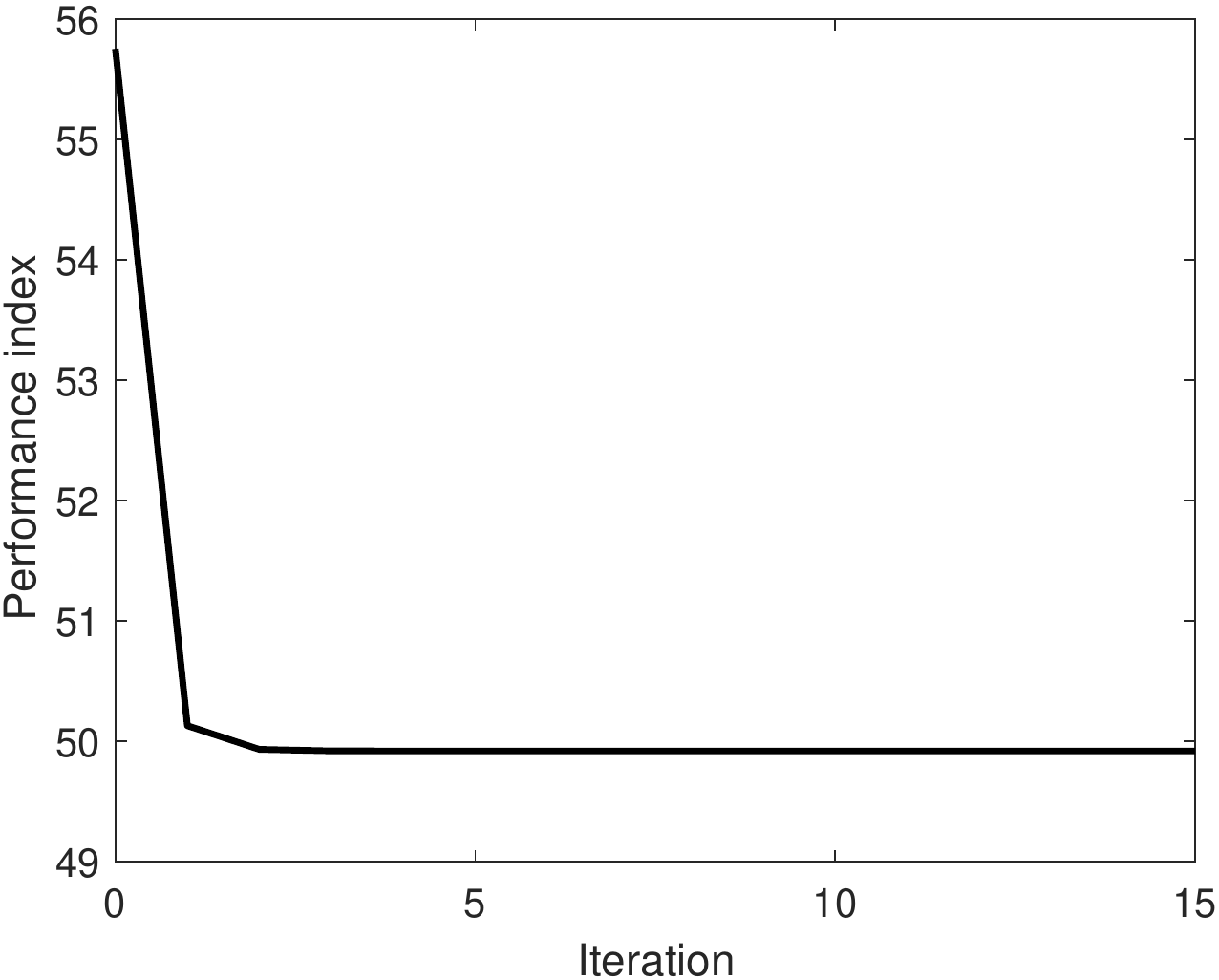}
		\caption{Convergence of performance index}\label{Fig:1}
	\end{figure}
	
	\begin{figure}
		\centering
		\includegraphics[scale=0.52]{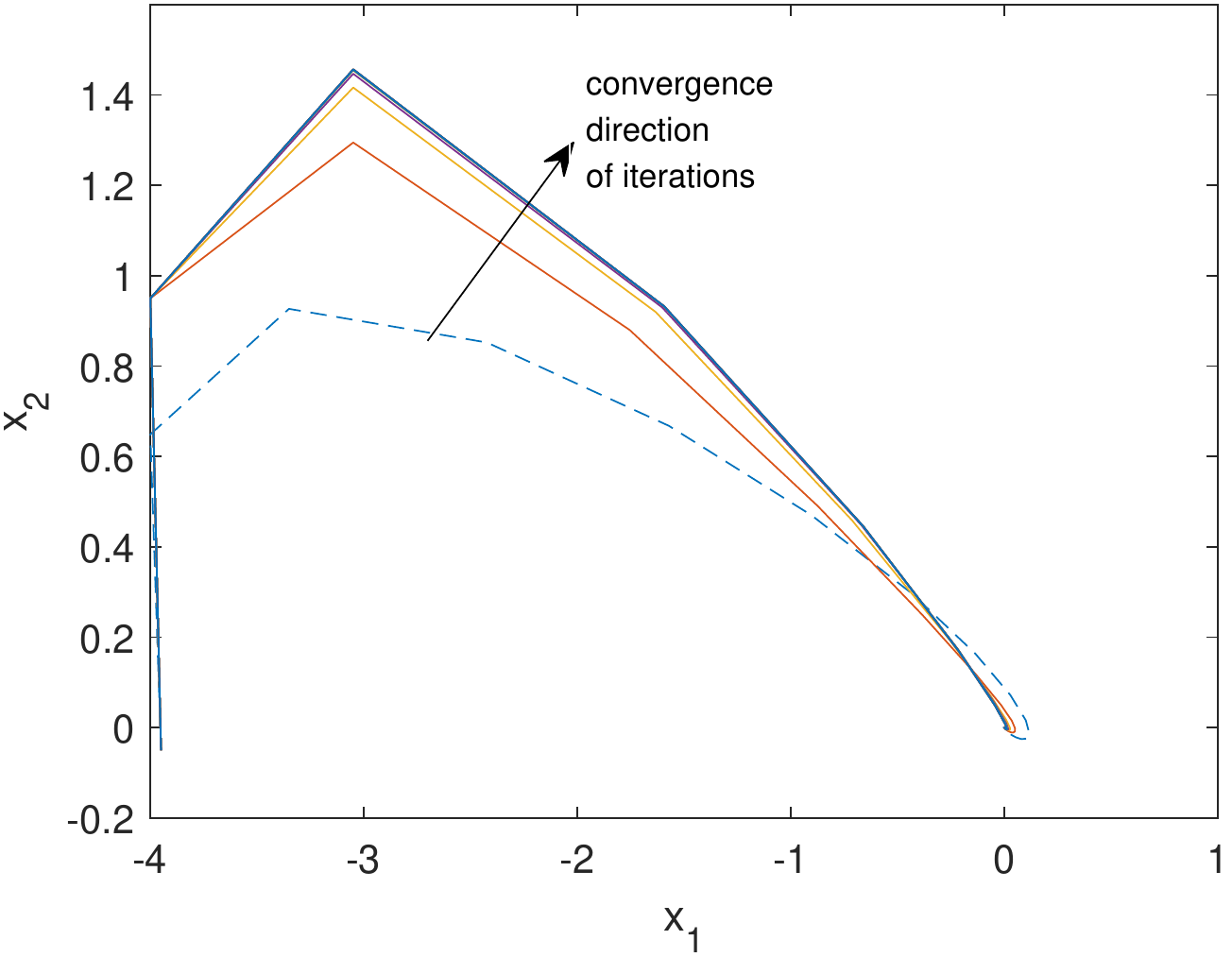}
		\caption{State trajectory of each iteration(dash curve is the initial feasible state trajectory)}\label{Fig:2}
	\end{figure}
	
	\subsubsection{Nonlinear case}
	
	Consider a nonlinear system:
	\begin{equation}
	x(k+1)=Ax(k)+g(x(k))+Bu(k),\nonumber
	\end{equation}
	
	with $x(k)=\begin{pmatrix}x^1(k)\\x^2(k)\end{pmatrix}$, $A=\begin{pmatrix}1&1\\0&1\end{pmatrix}$, $B=\begin{pmatrix}0\\1\end{pmatrix}$, $g(x(k))=(x^1(k)x^2(k)(1+sin(x^1(k)x^2(k))),0)^T$ and $x(0)=\begin{pmatrix}-3.95\\-0.05\end{pmatrix}$. The constraints for the system are 
	\begin{equation}
	\begin{pmatrix}-4\\-4\end{pmatrix}\le x(k)\le\begin{pmatrix}4\\4\end{pmatrix},~k\in\mathbb{N}\nonumber
	\end{equation}
	and
	\begin{equation}
	-1\le u(k)\le 1,~k.\in\mathbb{N}\nonumber
	\end{equation}
	
	The performance index to be minimized is $\sum_{k=0}^{\infty}\|x(k)\|_2^2+\|u(k)\|_2^2$. The initial feasible state and control sequence is generated by using an open-loop controller to drive the state to a small neighborhood of the origin and then using a stabilizing linear feedback controller. The prediction horizon is also chosen as 4. In Fig.~\ref{Fig:3} we present the evolution of the performance along iterative learning. In Fig.~\ref{Fig:4}, the state trajectories of iterations are shown. In Fig.~\ref{Fig:5} we show the trajectories of the last 17 iterations, from which the convergence of the trajectories is clear.
	
	\begin{figure}
		\centering
		\includegraphics[scale=0.52]{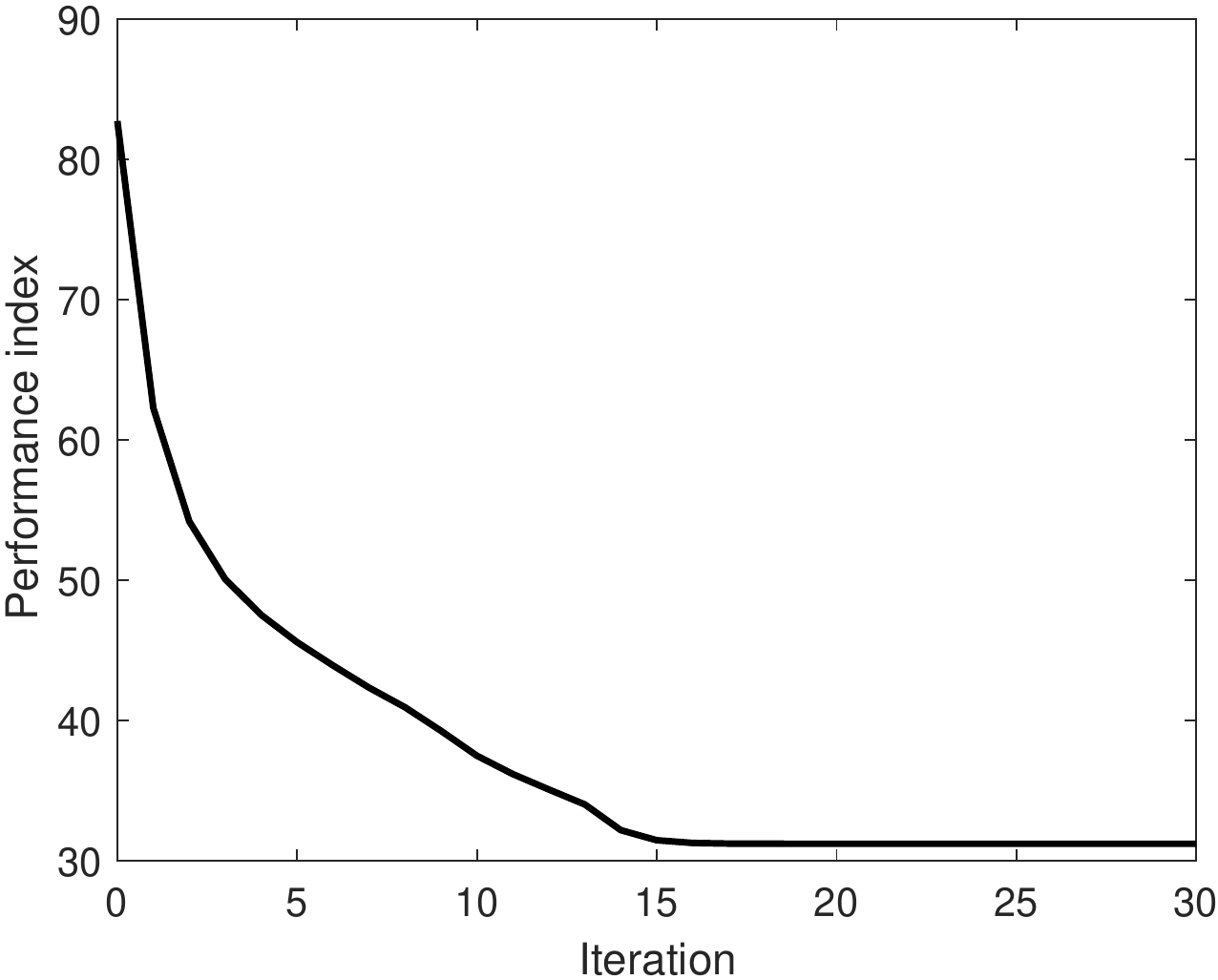}
		\caption{Convergence of performance index}\label{Fig:3}
	\end{figure}
	
	\begin{figure}
		\centering
		\includegraphics[scale=0.52]{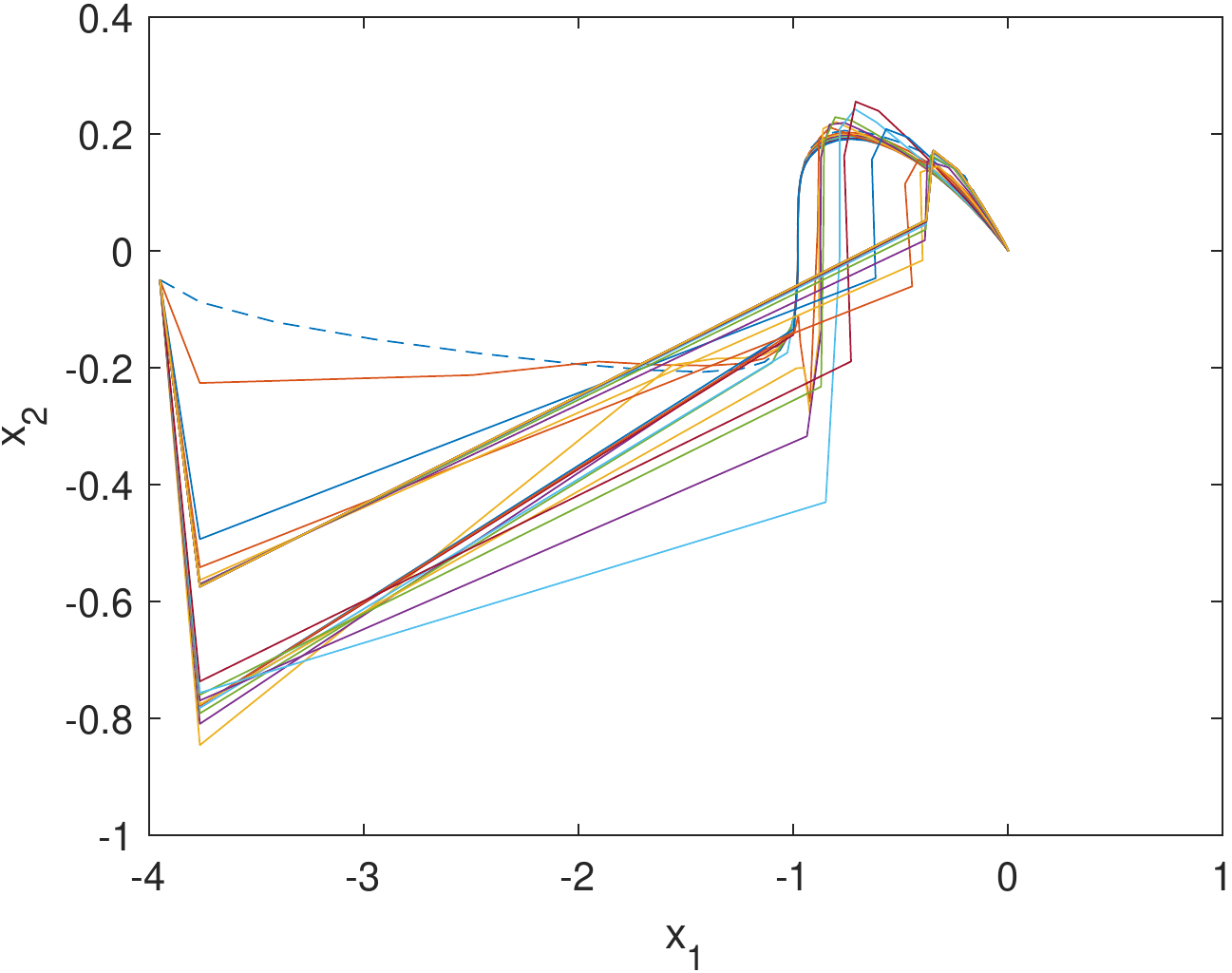}
		\caption{State trajectory of each iteration(dash curve is the initial feasible state trajectory)}\label{Fig:4}
	\end{figure}
	
	\begin{figure}
		\centering
		\includegraphics[scale=0.52]{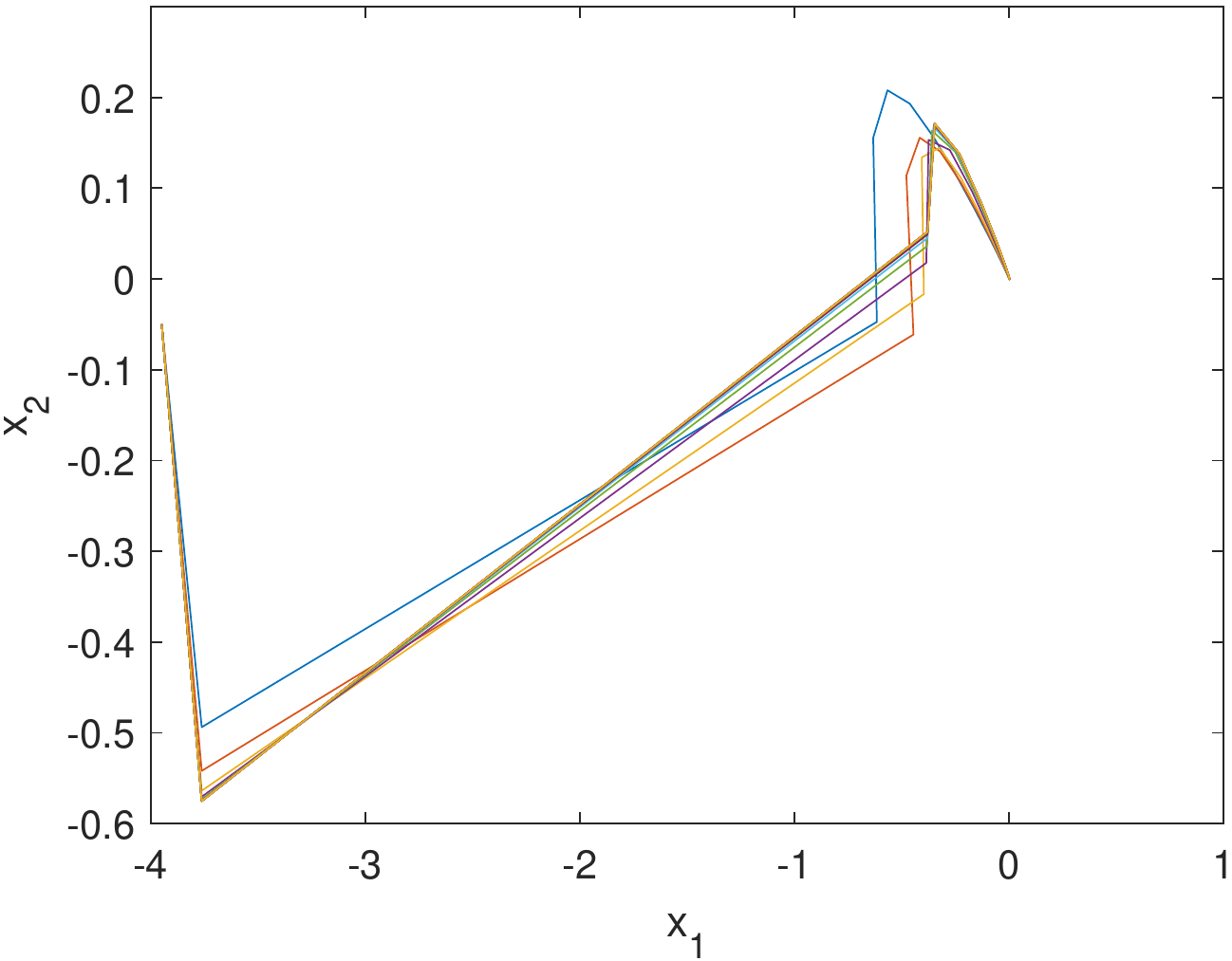}
		\caption{State trajectories of the last 17 iterations}\label{Fig:5}
	\end{figure}
	
	\subsection{Constrained tracking}
	\subsubsection{Linear agent}
	
	In this example, we consider the following agent:
	\begin{eqnarray}
	x(k+1)=Ax(k)+u(k),\nonumber
	\end{eqnarray}
	with $A=\begin{pmatrix}1&1\\0&1\end{pmatrix}$, $x(0)=\begin{pmatrix}0\\0\end{pmatrix}$.
	
	The constraints of the agent are
	
	\begin{equation}
	\begin{pmatrix}-4\\-4\end{pmatrix}\le x(k)\le\begin{pmatrix}5\\5\end{pmatrix},~k\in\mathbb{N}\nonumber
	\end{equation}
	and
	\begin{equation}
	\begin{pmatrix}-1\\-1\end{pmatrix}\le u(k)\le \begin{pmatrix}1\\1\end{pmatrix},~k\in\mathbb{N}.\nonumber
	\end{equation}
	
	The target trajectory is a square with width 4, center $(4,4)$ and period $T=16$. The prediction horizon is chosen as $N=4$. Note that due to the state and input constraints, perfect tracking is impossible. The optimal reachable trajectory is defined by the following optimization problem:
	\begin{equation}
	\min_{x_0,u_0,\ldots,u_{T-1}}\sum_{k=0}^{T-1}\|x(k)-r(k)\|_2^2\nonumber
	\end{equation} 
	subject to
	\begin{eqnarray}
	x(k+1)&=&Ax(k)+u(k),\nonumber\\
	x(0)&=&x_0,\nonumber\\
	x(T)&=&x(0),\nonumber\\
	\begin{pmatrix}-4\\-4\end{pmatrix}&\le& x(k)\le\begin{pmatrix}5\\5\end{pmatrix},\nonumber\\
	\begin{pmatrix}-1\\-1\end{pmatrix}&\le& u(k)\le \begin{pmatrix}1\\1\end{pmatrix},\nonumber\\
	k&=&0,\ldots,T-1.\nonumber
	\end{eqnarray}
	This optimal reachable trajectory will not be implemented in the control algorithm. It is labeled in Fig.~\ref{Fig:7} by using yellow cross. The tracking error $\sum_{k=0}^{\infty}\|x(k)-r(k)\|_2^2$ is infinite. So we consider the average tracking error as $\lim_{T\to\infty}\frac{1}{T}\sum_{k=0}^{T-1}\|x(k)-r(k)\|_2^2$. The initial feasible trajectory is a given periodic trajectory starts from and ends at the origin with period $T=16$. In Fig.~\ref{Fig:6}, we show the performance index of each iteration. In Fig.~\ref{Fig:7}, the state trajectory of each iteration is shown. We can see that the trajectory converges after 3 iteration. As expected, the trajectory converges to the optimal reachable trajectory.

	\begin{figure}
		\centering
		\includegraphics[scale=0.52]{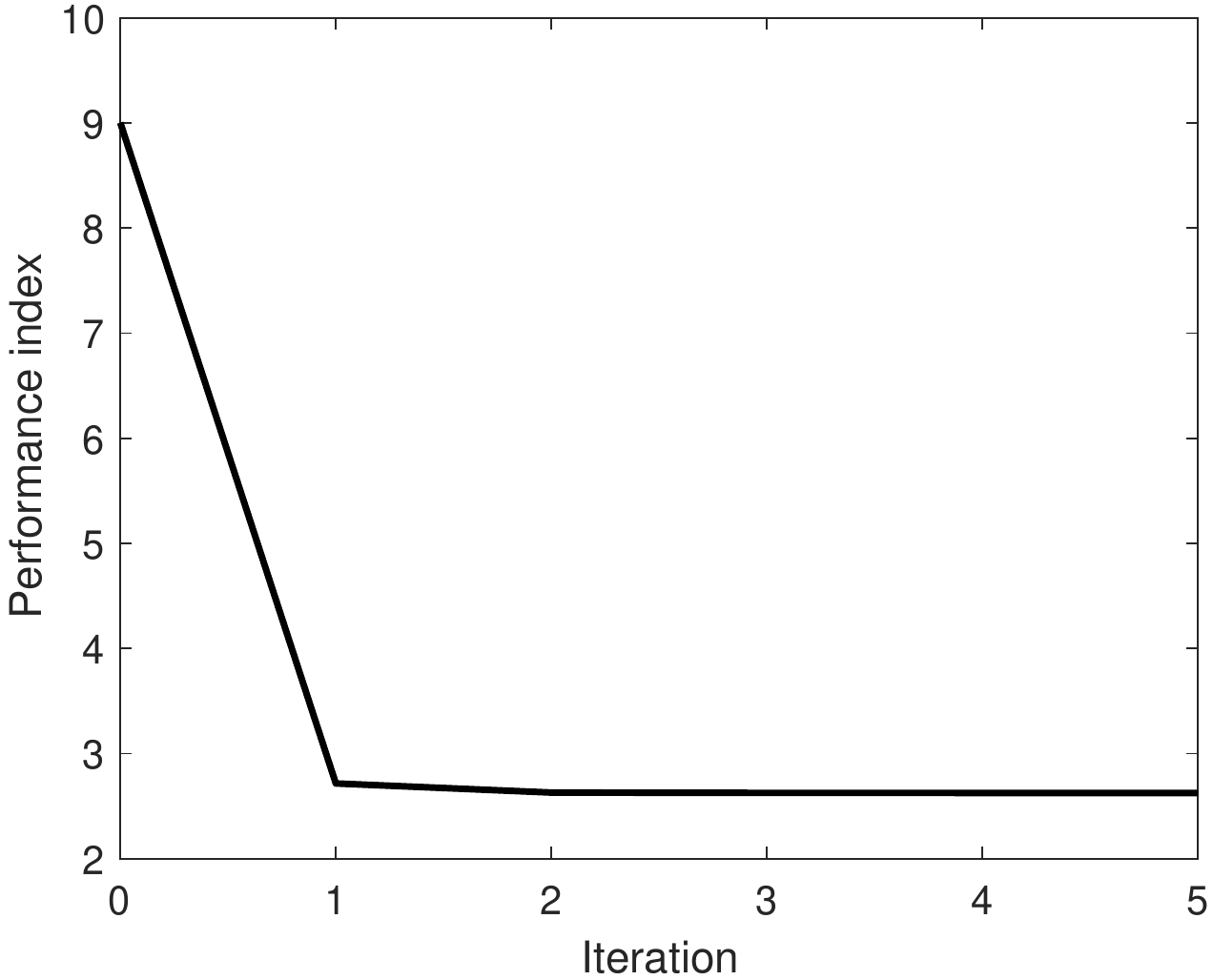}
		\caption{Convergence of performance index}\label{Fig:6}
	\end{figure}
	
	\begin{figure}
		\centering
		\includegraphics[scale=0.52]{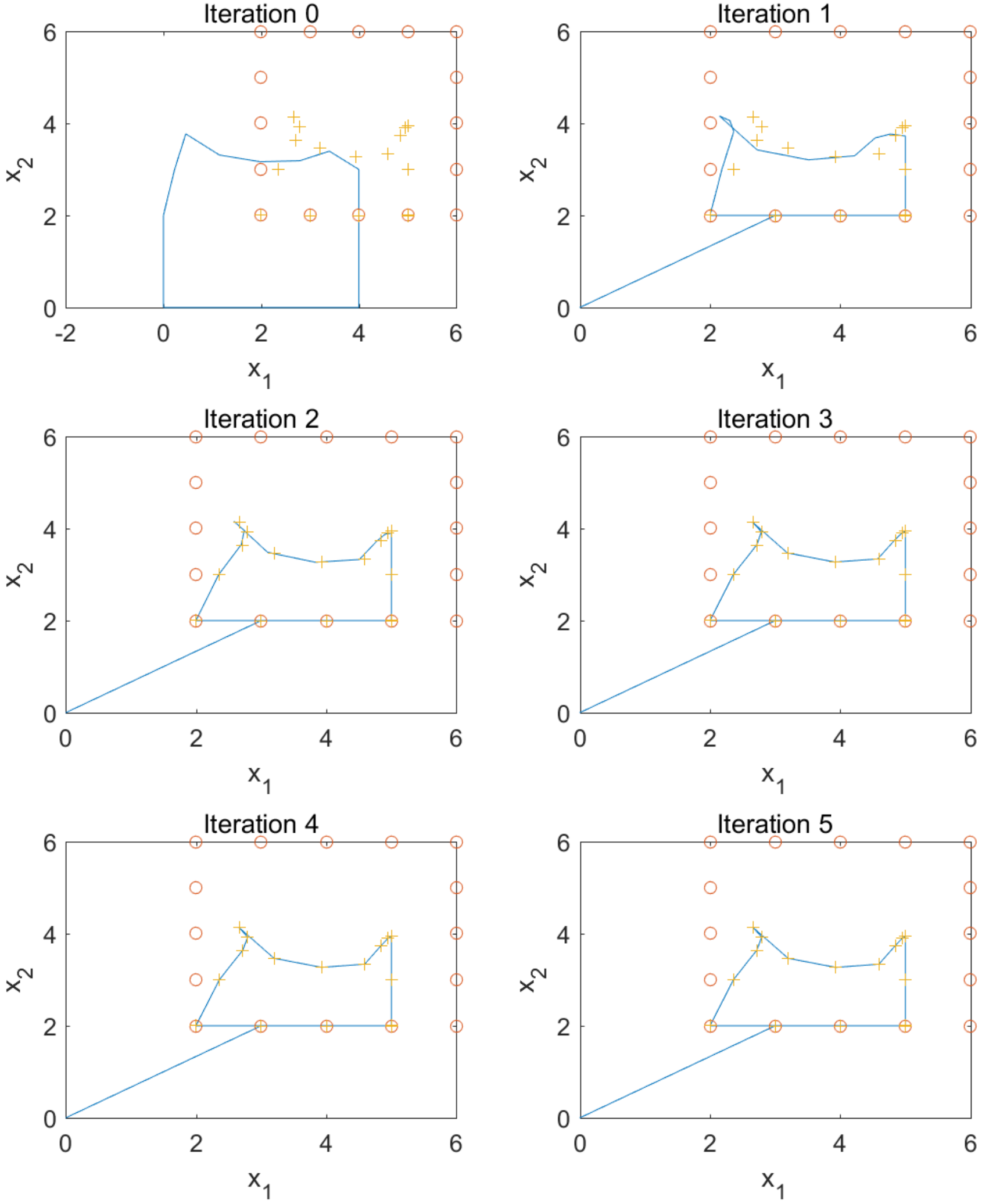}
		\caption{Trajectory of each iteration (circle denotes target trajectory, cross denotes optimal reachable trajectory)}\label{Fig:7}
	\end{figure}
	
	\subsubsection{Nonlinear vehicle}
	
	In this example we consider a simplified nonlinear vehicle model:
	\begin{eqnarray}
	\dot{x}&=&vcos(\theta),\nonumber\\
	\dot{y}&=&vsin(\theta),\nonumber\\
	\dot{v}&=&a,\nonumber\\
	\dot{\theta}&=&\omega,\nonumber
	\end{eqnarray}
	where $(x,y)$ denotes the position of this vehicle, $v$ is the velocity, $\theta$ is the direction of the velocity, $a$ is the acceleration and $\omega$ is the angle velocity. Suppose that the control input is the acceleration and the angle velocity and the constraints are $-15\text{m}/\text{s}^2\le a\le 15\text{m}/\text{s}^2$ and $-12\text{rad}/\text{s}\le w\le12\text{rad}/\text{s}$. This system is discretized with sampling time interval $\delta=0.1$s. The initial condition of this vehicle is set as $(0,0,0,0)^T$. The target trajectory is a circle with radius of 5 meter, center $(6,6)$ and period $T=4$s. The prediction horizon is chosen as $N=10$. In Fig.~\ref{Fig:8}, we present the average tracking error of each iteration and we can see that it converges after 10 iterations. In Fig.~\ref{Fig:9}, we present the initial feasible trajectory and the trajectory of the first 11 iterations. One can observe that, though the initial trajectory is a totally different periodic trajectory, after a few iterations, the vehicle can learn to approach to the target trajectory.
	
	\begin{figure}
		\centering
		\includegraphics[scale=0.52]{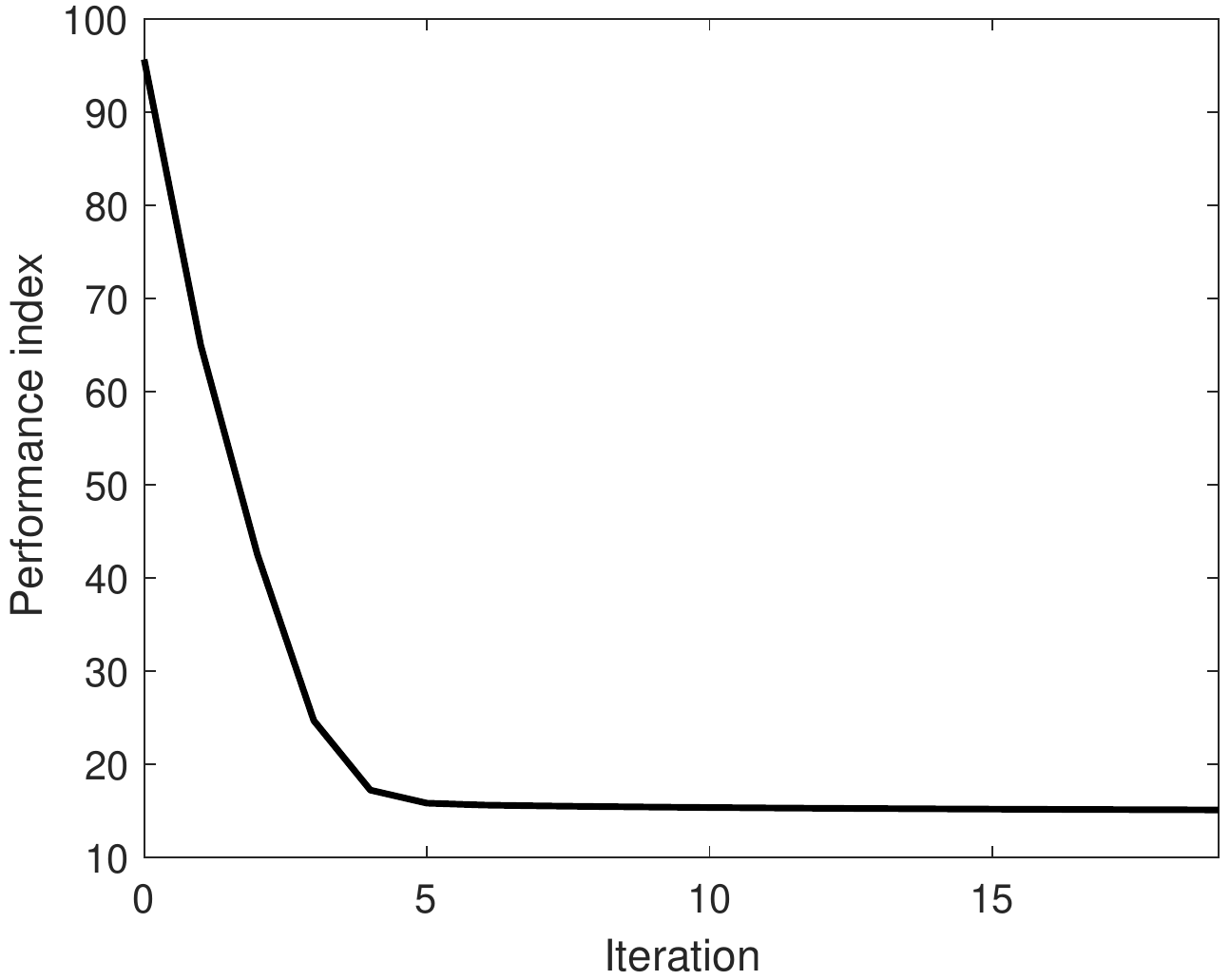}
		\caption{Convergence of performance index}\label{Fig:8}
	\end{figure}
	
	\begin{figure}
		\centering
		\includegraphics[scale=0.6]{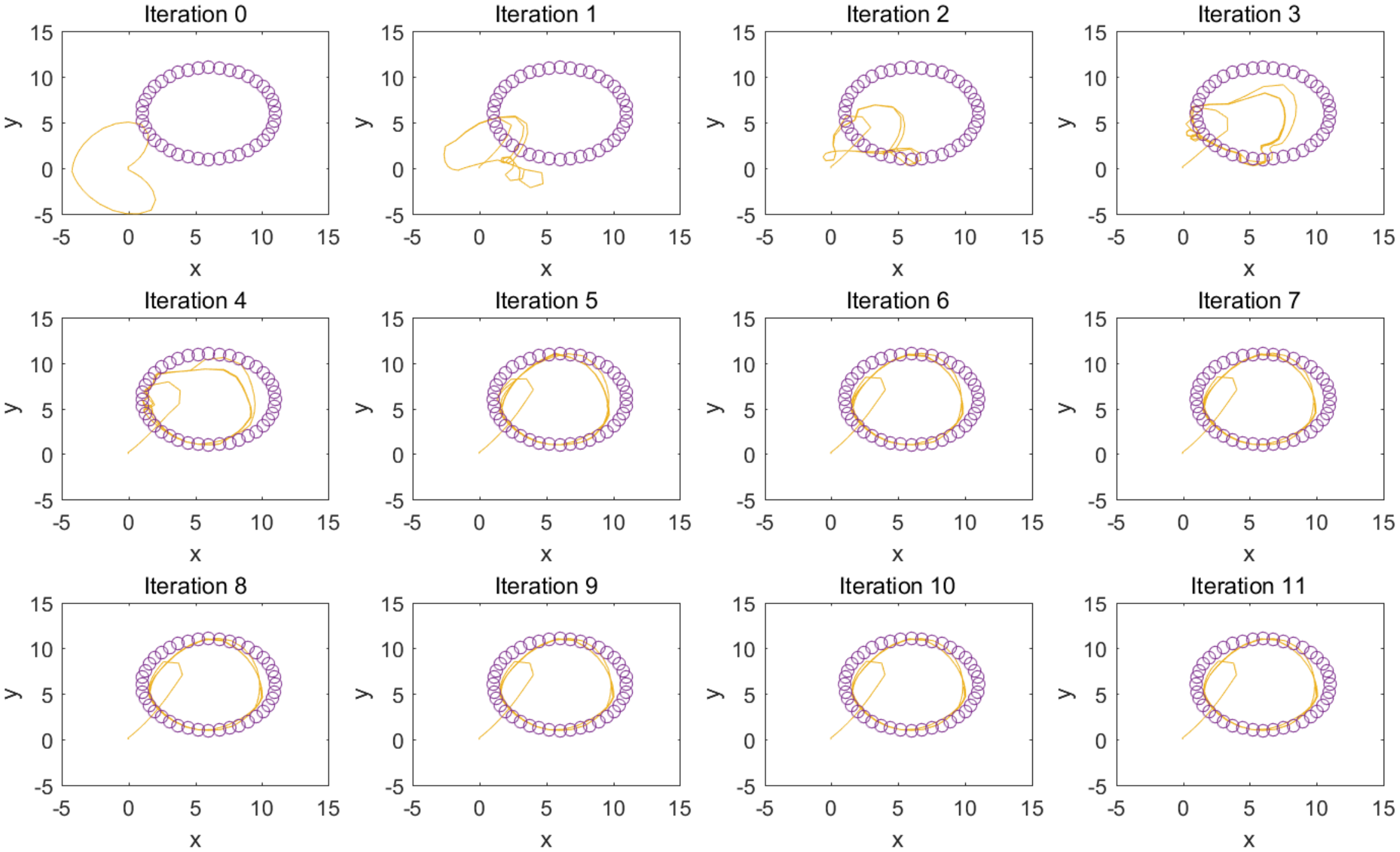}
		\caption{Trajectory of each iteration (circle denotes target trajectory)}\label{Fig:9}
	\end{figure}
	
	\subsection{Consecutive-competitive reactions}
	In the next two examples, we use a nonlinear model of an isothermal chemical reactor with consecutive-competitive reactions \cite{Angeli12}:
	\begin{eqnarray}
	P_0+B\to P_1,\nonumber\\
	P_1+B\to P_2.\nonumber
	\end{eqnarray}
	
	The dynamic model is given by
	\begin{eqnarray}
	\dot{x}_1&=&u_1-x_1-\sigma_1x_1x_2,\nonumber\\
	\dot{x}_2&=&u_2-x_2-\sigma_1x_1x_2-\sigma_2x_2x_3,\nonumber\\
	\dot{x}_3&=&-x_3-\sigma_1x_1x_2-\sigma_2x_2x_3,\nonumber\\
	\dot{x}_4&=&-x_4+\sigma_2x_2x_3,\nonumber
	\end{eqnarray}
	where $x_1,~x_2,~x_3$ and $x_4$ are the concentrations of $P_0$, $B$, $P_1$ and $P_2$ respectively, while $u_1$ and $u_2$ are inflow rates of $P_0$ and $B$, which are the manipulated variables. The parameters $\sigma_1$ and $\sigma_2$ have values 1 and 0.4, respectively. The model is discretized by sampling time interval $0.1$s and the prediction horizon is chosen as $N=5$.

	The time average value of $u_1$ is set as
	\begin{equation}
	Av[u_1]\subset[0,1],\nonumber
	\end{equation}
	and a hard constraint $0\le u_1\le 5$ is also enforced. The control objective is to maximize the average amount of $P_1$ in the effluent flow ($l(x,u)=-x_3$). The steady state of the system is given by $x_s=(0.3874,1.5811,0.3752,0.2373)^T$ and $u_s=(1,2.4310)^T$. The initial feasible trajectory is generated by using an open-loop controller to drive the system state to the steady state and then using $u_s$ to maintain the steady state.
	
	\subsubsection{Pure economic cost}
	We first test the proposed ILEMPC algorithm with a pure economic cost $l(x,u)=-x_3$ for 15 iterations. Only the first 5 iterations are shown in Fig.~\ref{Fig:10}-\ref{Fig:13} for clarity. In Fig.~\ref{Fig:10} and \ref{Fig:11} we present the closed-loop state trajectory and $x_3$ of each iteration. As we can observe, though the initial trajectory is convergent, to obtain more average amount of $P_1$, the controller gradually learns that to make $x_3$ oscillate around $x_3=0.4$ is better. In Fig.~\ref{Fig:12} and \ref{Fig:13} we present the control input and $u_1$ of each iteration. Fig.~\ref{Fig:14} shows that though $u_1$ keeps oscillating, the average of $u_1$ gradually converges to the upper bound of the given set for each iteration. Finally, Fig.~\ref{Fig:15} shows that the average performance is improved along the learning process.
	
	\begin{figure}
		\centering
		\includegraphics[scale=0.5]{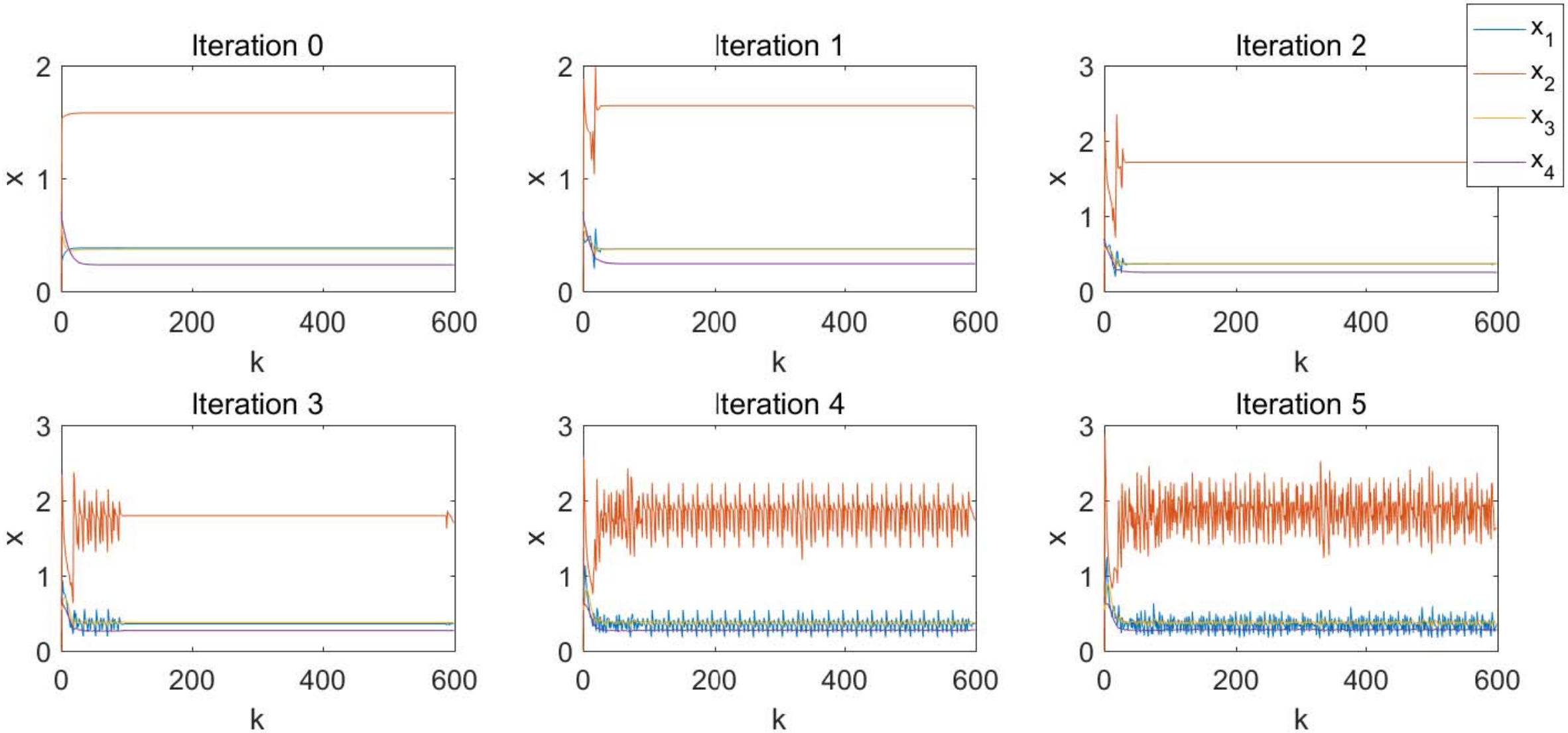}
		\caption{System state with economic cost}\label{Fig:10}
	\end{figure}
	
	\begin{figure}
		\centering
		\includegraphics[scale=0.5]{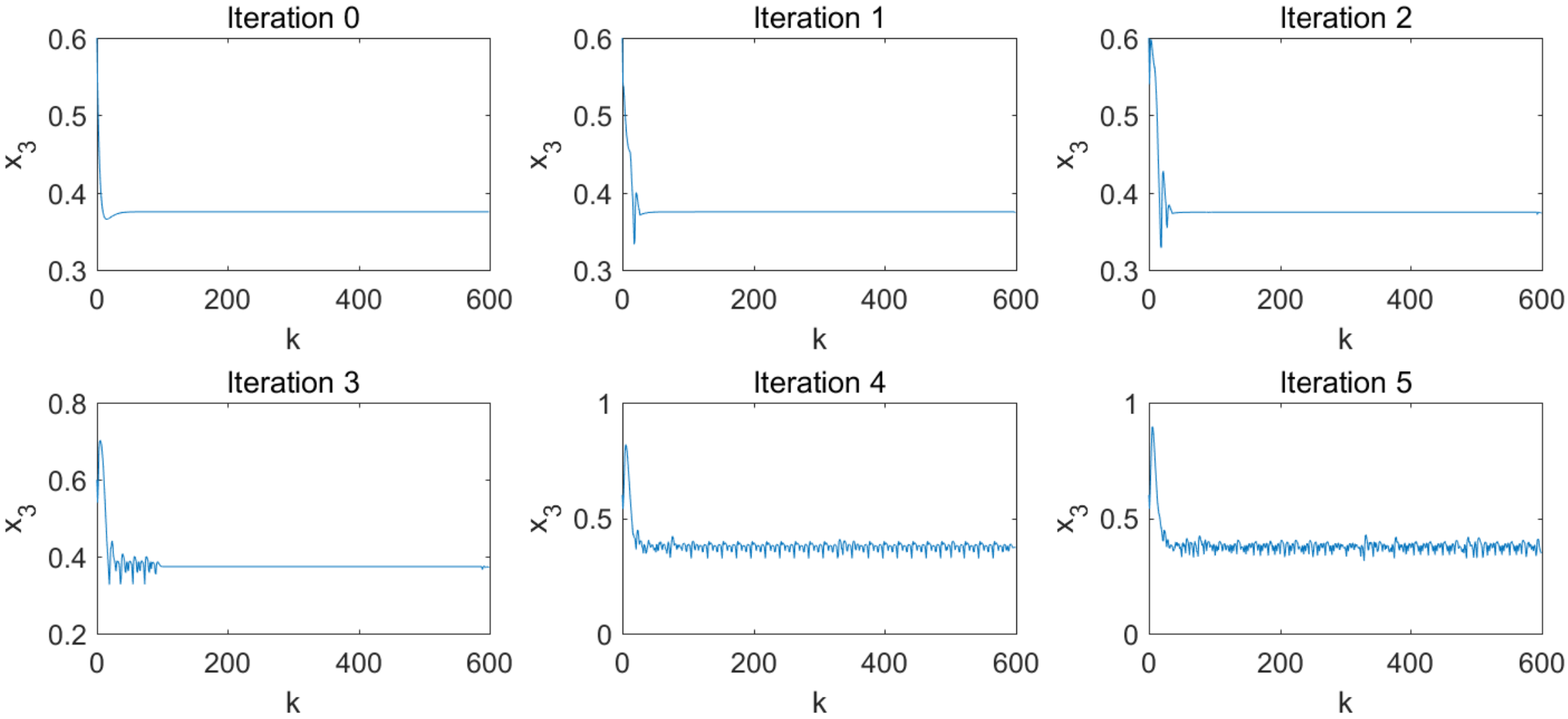}
		\caption{$x_3$ with economic cost}\label{Fig:11}
	\end{figure}
	
	\begin{figure}
		\centering
		\includegraphics[scale=0.5]{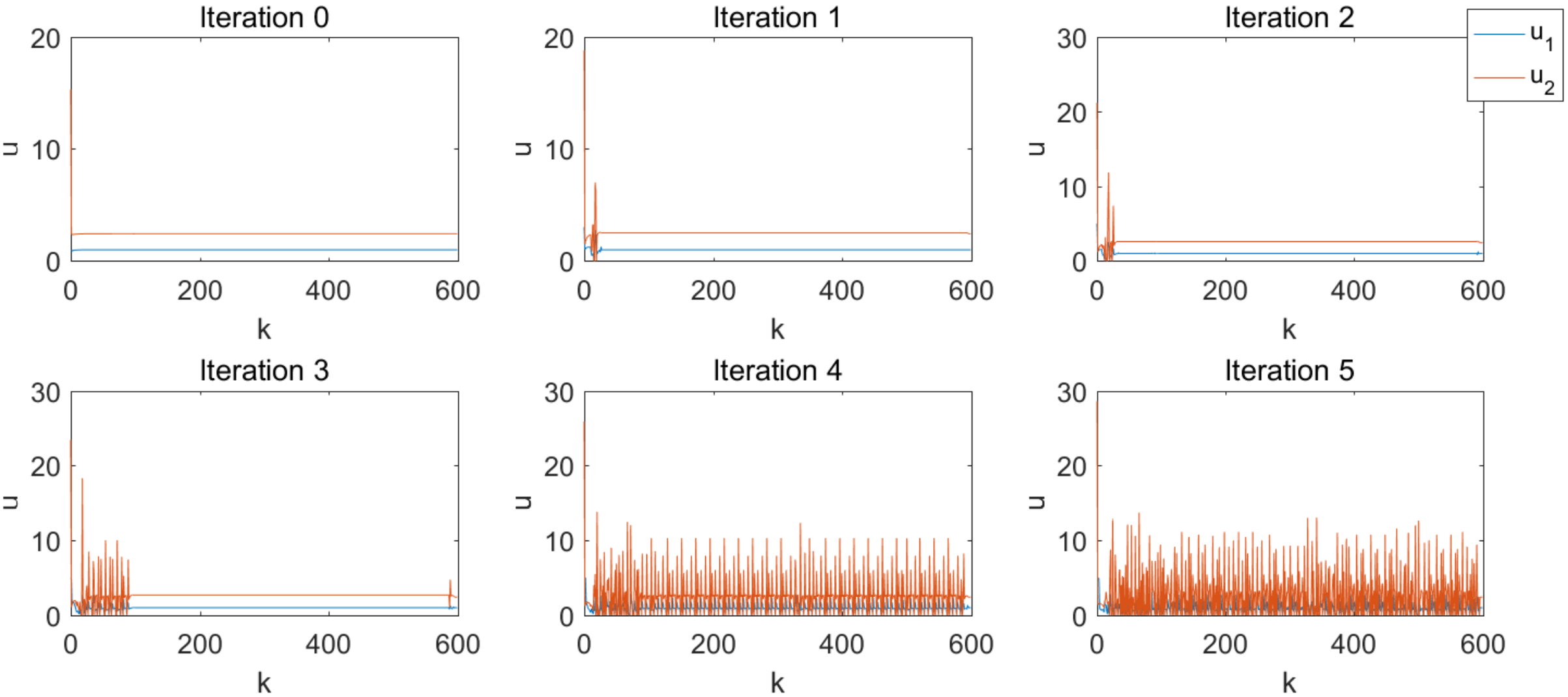}
		\caption{Control input with economic cost}\label{Fig:12}
	\end{figure}
	
	\begin{figure}
		\centering
		\includegraphics[scale=0.5]{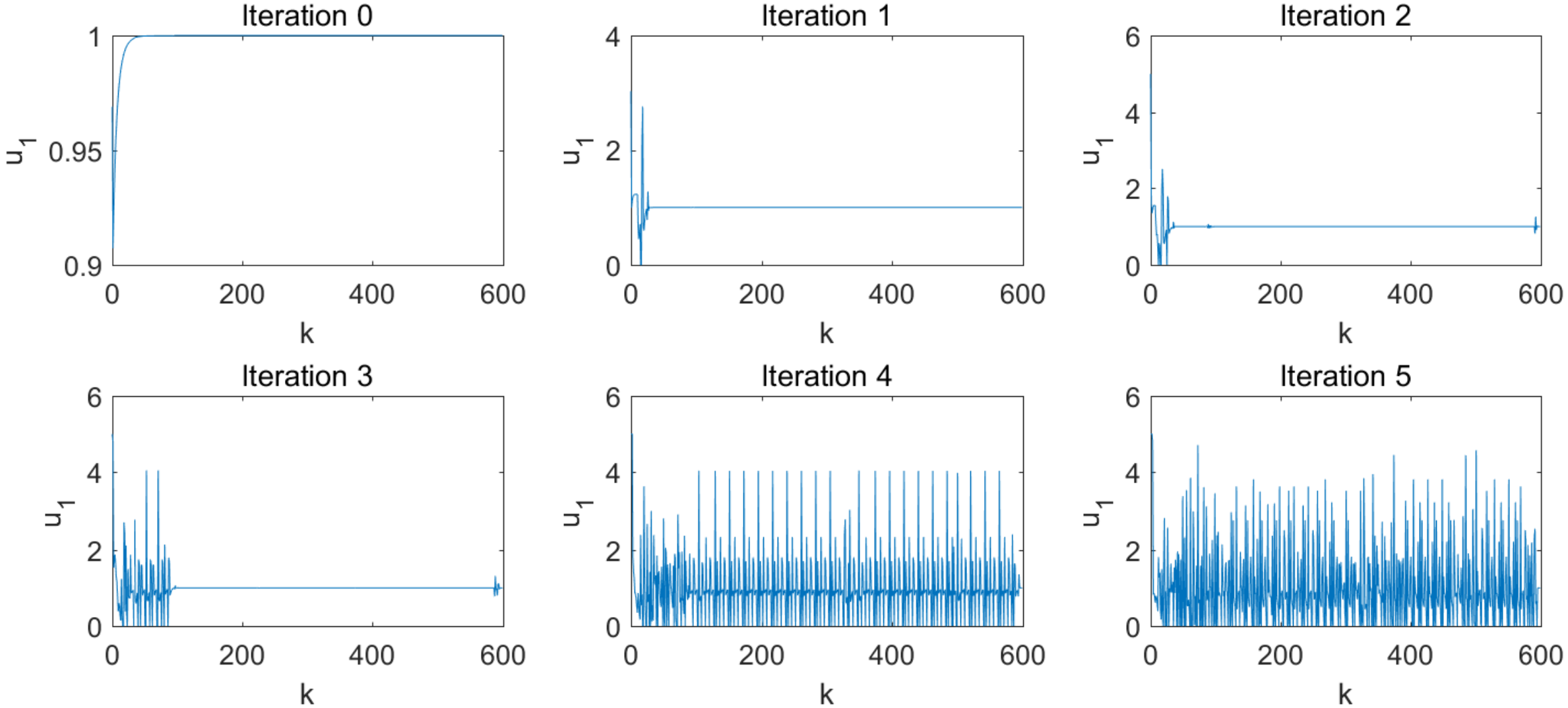}
		\caption{$u_1$ with economic cost}\label{Fig:13}
	\end{figure}
	
	\begin{figure}
		\centering
		\includegraphics[scale=0.5]{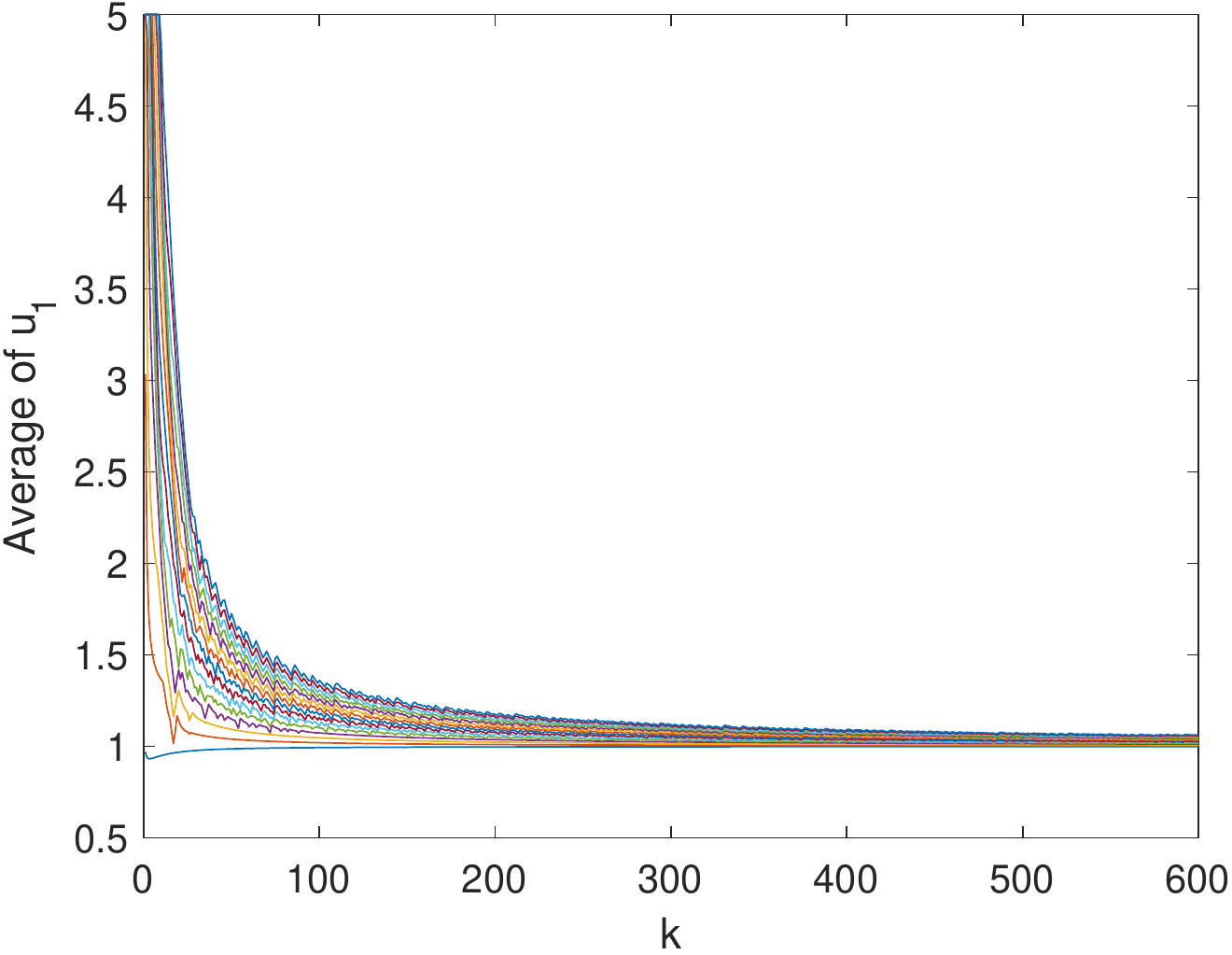}
		\caption{Average amount of $u_1$}\label{Fig:14}
	\end{figure}
	
	\begin{figure}
		\centering
		\includegraphics[scale=0.5]{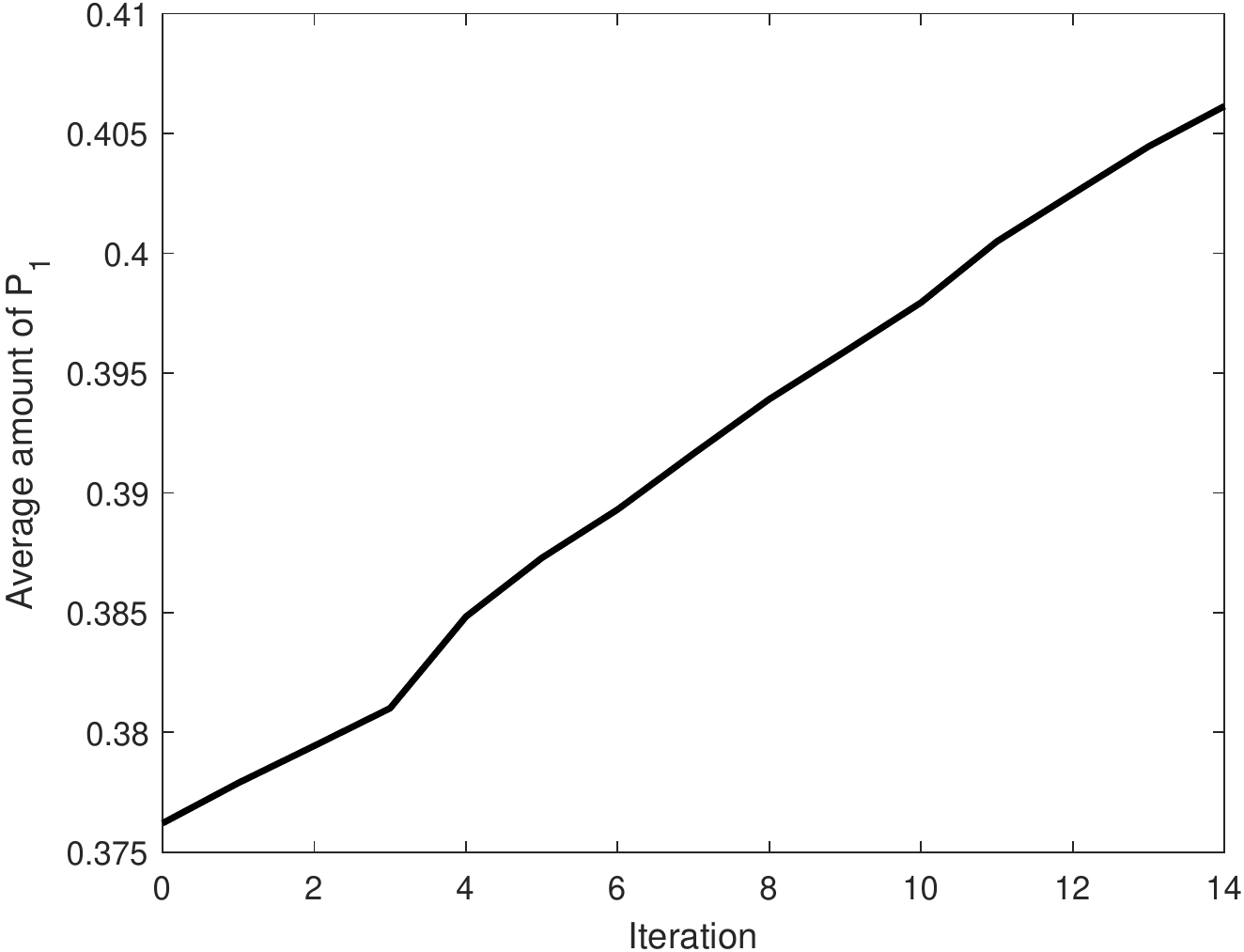}
		\caption{Average amount of $P_1$}\label{Fig:15}
	\end{figure}
	
	\subsubsection{Convexified economic cost}
	We then test the proposed ILEMPC algorithm with a convexified economic cost $l(x,u)=-x_3+\frac{1}{2}(\|x-x_s\|^2_Q+\|u-u_s\|^2_R)~,Q=0.36I_4,~R=0.002I_2$, which makes the dissipative assumption hold. In this case, the average amount of $P_1$ for each iteration will be the same since by Theorem \ref{converge}, the state trajectory of each iteration also converges to the steady state. Therefore, we compare $\sum_{k=0}^{\infty}(l(x_j(k),u_j(k))-l(x_s,u_s))$ for each iteration. The algorithm is tested for 15 iterations. Only the first 5 iterations are shown in Fig.~\ref{Fig:16}-\ref{Fig:19} for clarity. In Fig.~\ref{Fig:16} and \ref{Fig:17} we present the closed-loop state trajectory and $x_3$ of each iteration. In Fig.~\ref{Fig:18} and \ref{Fig:19} we present the control input and $u_1$ of each iteration. Fig.~\ref{Fig:20} shows that the average of $u_1$ gradually converges to the upper bound of the given set for each iteration. Finally, Fig.~\ref{Fig:21} shows that the transient performance is improved along the learning process.
	
	\begin{figure}
		\centering
		\includegraphics[scale=0.5]{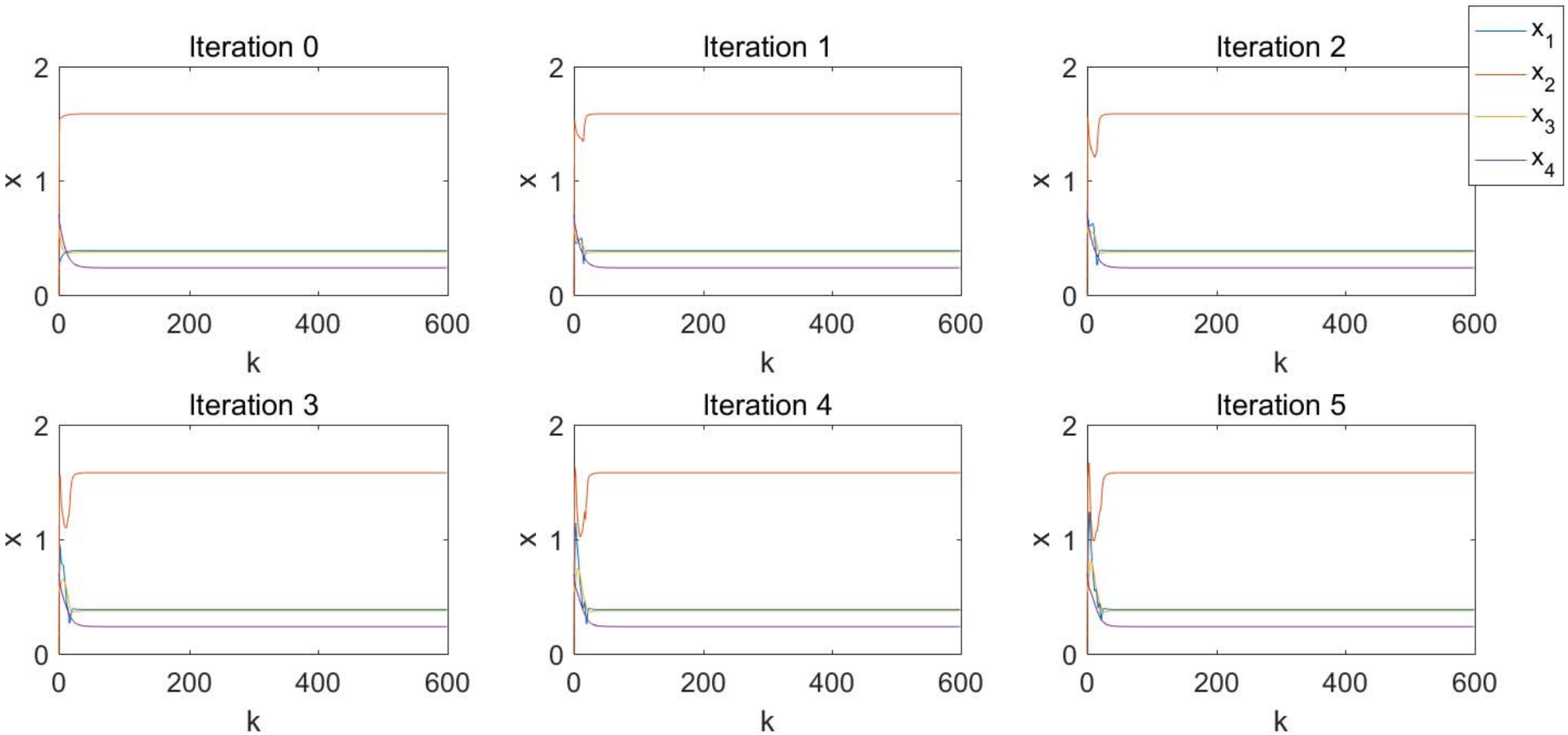}
		\caption{System state with convexified cost}\label{Fig:16}
	\end{figure}
	
	\begin{figure}
		\centering
		\includegraphics[scale=0.5]{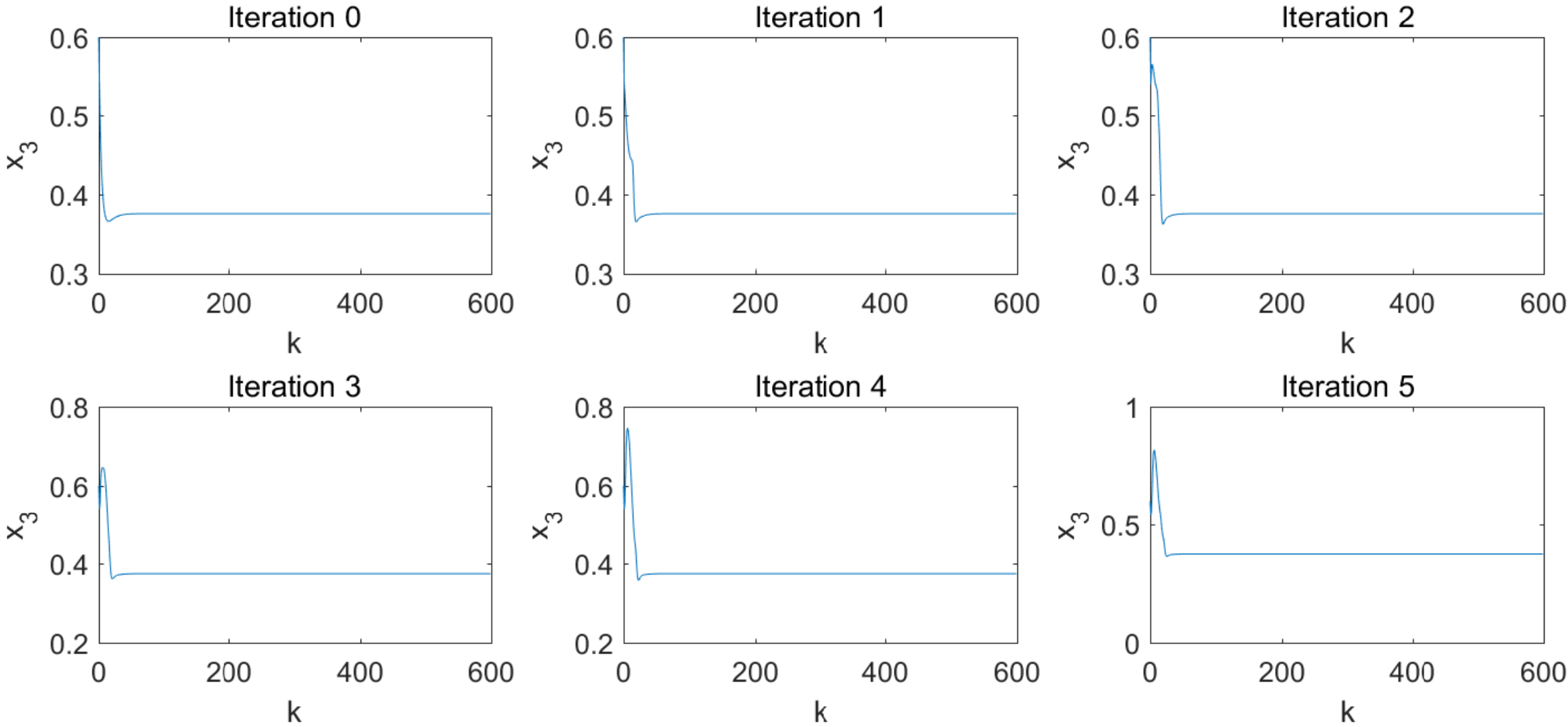}
		\caption{$x_3$ with convexified cost}\label{Fig:17}
	\end{figure}
	
	\begin{figure}
		\centering
		\includegraphics[scale=0.5]{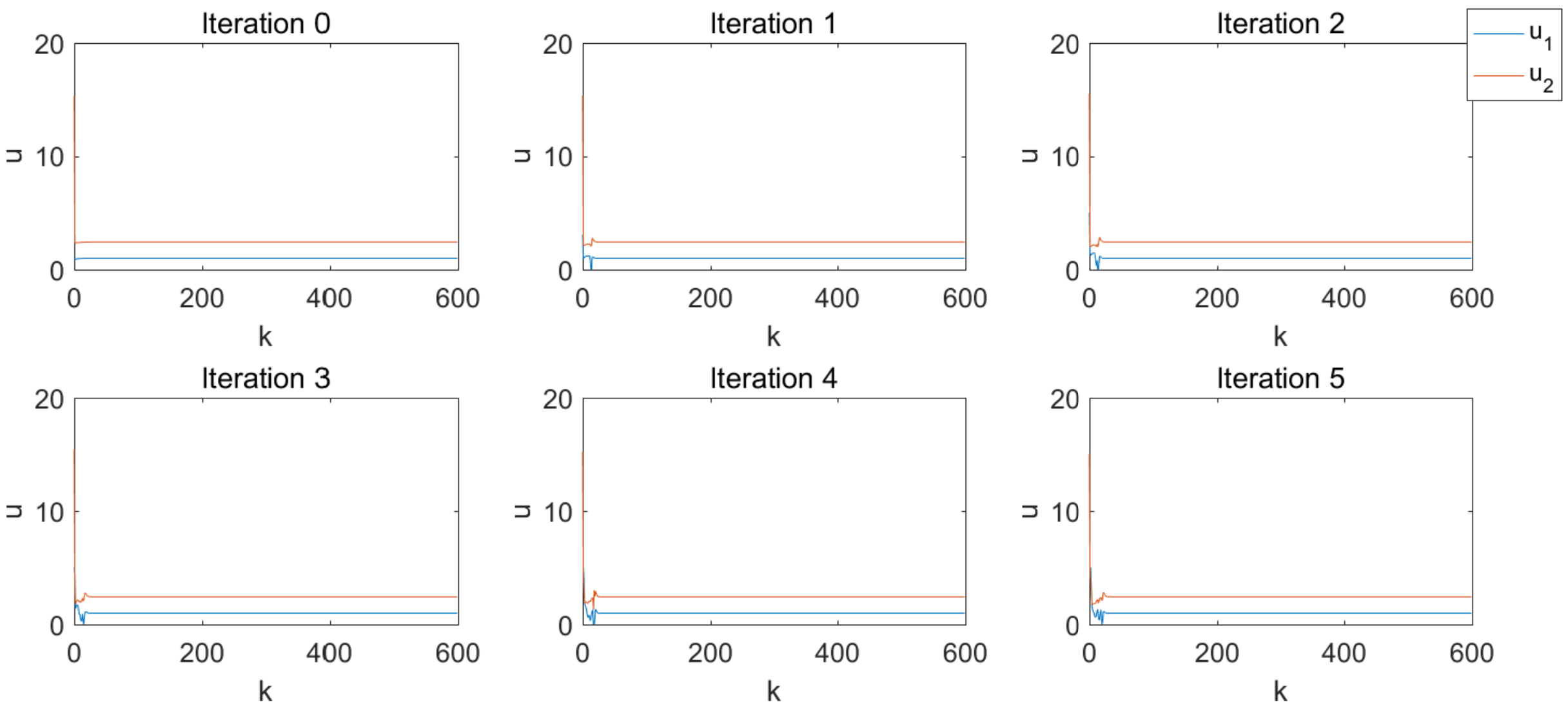}
		\caption{Control input with convexified cost}\label{Fig:18}
	\end{figure}
	
	\begin{figure}
		\centering
		\includegraphics[scale=0.5]{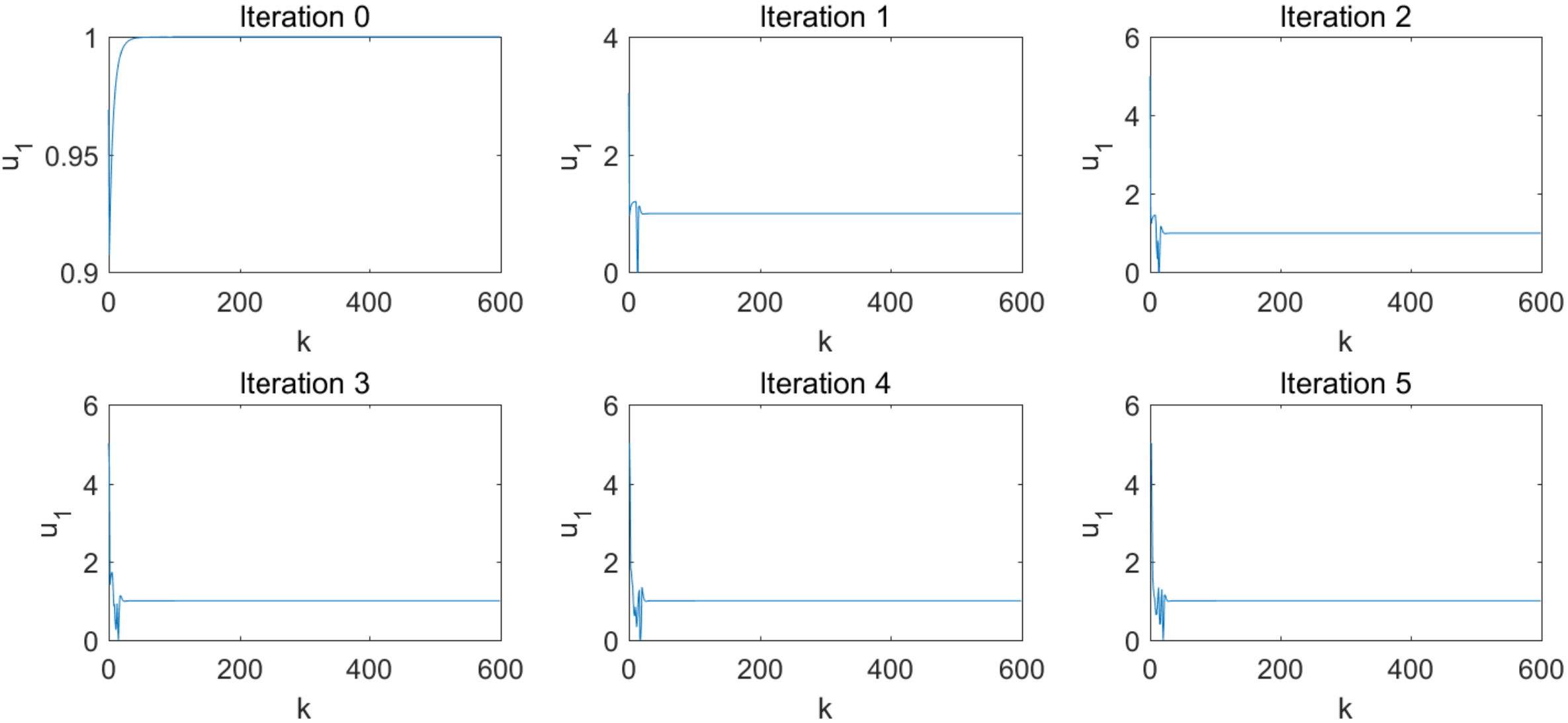}
		\caption{$u_1$ with convexified cost}\label{Fig:19}
	\end{figure}
	
	\begin{figure}
		\centering
		\includegraphics[scale=0.5]{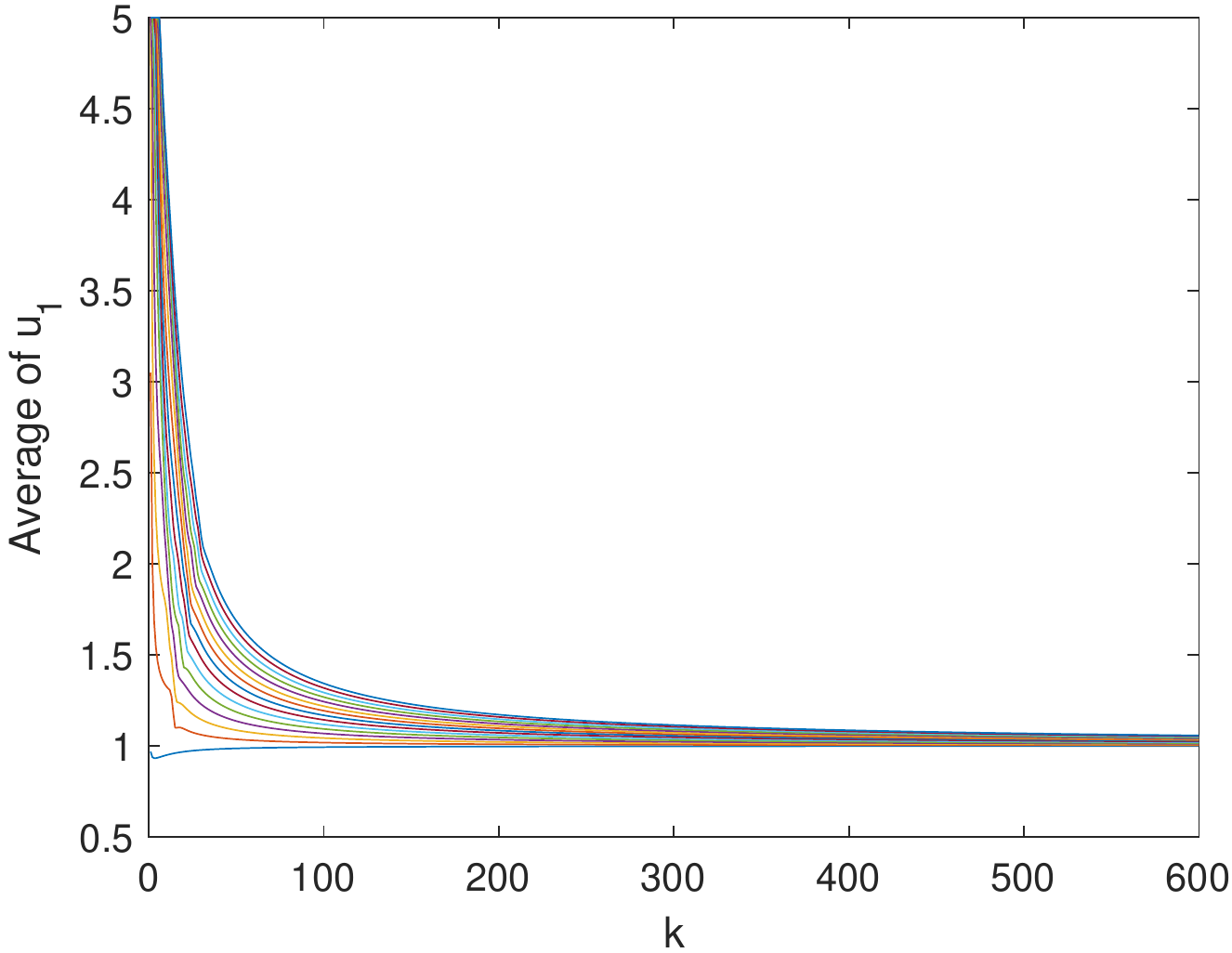}
		\caption{Average amount of $u_1$}\label{Fig:20}
	\end{figure}
	
	\begin{figure}
		\centering
		\includegraphics[scale=0.5]{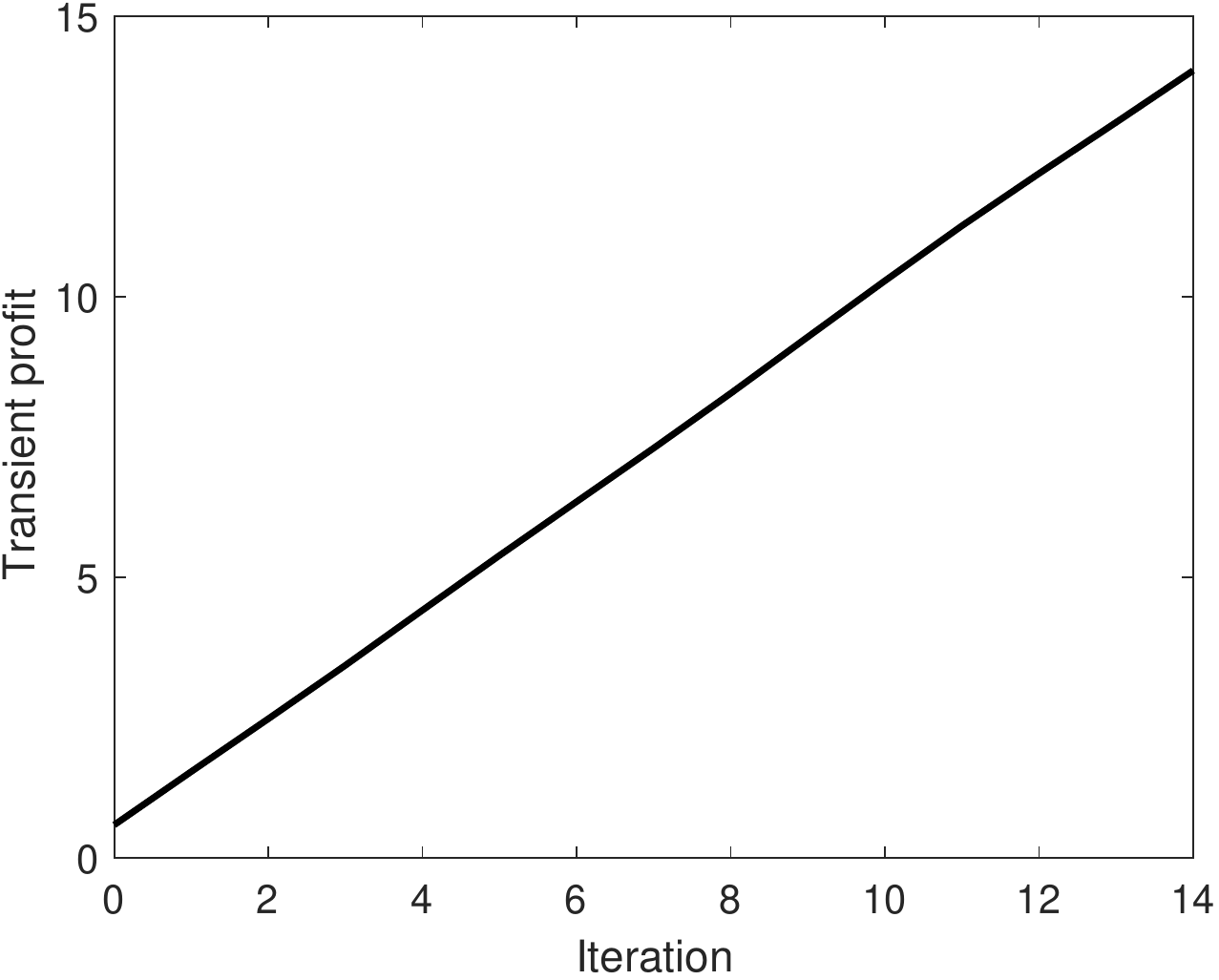}
		\caption{Average amount of $P_1$}\label{Fig:21}
	\end{figure}
	
	\section{Conclusion}\label{conclusion}
	
	In this paper, a learning-based economic model predictive control algorithm for iterative tasks has been proposed. The main features of the proposed control algorithm are: 1) it is capable of exploiting exploit information from the last execution to improve the closed-loop performance; 2) the interested performance index is not limited to tracking error but could contain general economic cost of the plant operation. We have proved that at each iteration, the performance index to be optimized will be no worse than that of the previous iteration. For the stabilization problem, we have proved that under the dissipative assumption, the stability of the initial feasible trajectory is preserved. After that, under some assumptions on the uniqueness of the optimum, we have proved that if the closed-loop trajectory converges to a steady state trajectory, then it is the $N$-receding-horizon optimal trajectory. The proposed ILEMPC has been tested on constrained stabilization problems and unreachable tracking problems for both linear and nonlinear systems and a nonlinear isothermal chemical reator model. The effectiveness of the proposed algorithm has been verified.
	
	\bibliographystyle{IEEEtrans}
	\bibliography{mybib}

\end{document}